%% file: main.tex
\documentclass[a4paper,UKenglish,cleveref, autoref, thm-restate]{lipics-v2021}

\usepackage{sidecap}
\usepackage{booktabs}
\usepackage{makecell}
\usepackage{multirow}
\usepackage{array}
\usepackage{arydshln}
\usepackage{enumitem}

\usepackage{makecell} 
\usepackage{multirow}
\usepackage{arydshln} 

\usepackage{footmisc} 

\usepackage[normalem]{ulem}

\theoremstyle{definition}
\newtheorem{definition2}[theorem]{Definition}  
\crefname{definition2}{Definition}{Definitions} 

\title{Algorithmic hardness of the partition function for nucleic acid strands}

\titlerunning{Algorithmic hardness of the partition function for nucleic acid strands} 

\author{Gwendal Ducloz} {Hamilton Institute and Department of Computer Science, Maynooth University, Ireland \and École Normale Supérieure de Lyon, France}{gwendal.ducloz@ens-lyon.fr}{}{}

\author{Ahmed Shalaby}{Hamilton Institute and Department of Computer Science, Maynooth University, Ireland \and \url{http://www.shalabycave.com}}{ahmed.shalaby.2023@mumail.ie}{https://orcid.org/0009-0009-3054-4328}{}

\author{Damien Woods}{Hamilton Institute and Department of Computer Science, Maynooth University, Ireland \and \url{https://dna.hamilton.ie/woods}}{damien.woods@mu.ie}{https://orcid.org/0000-0002-0638-2690}{}

\authorrunning{G. Ducloz, A. Shalaby, and D. Woods} 

\Copyright{Gwendal Ducloz, Ahmed Shalaby, and Damien Woods} 


\category{} 

\relatedversion{} 

\funding{Work carried out while GD was on internship at Maynooth University. Supported by Science Foundation Ireland (SFI) under grant number 20/FFP-P/8843, European Research Council (ERC) under the European Union's Horizon 2020 research and innovation programme (grant agreement No 772766, Active-DNA project), and European Innovation Council (EIC), No 101115422, DISCO project. Views and opinions expressed are however those of the author(s) only and do not necessarily reflect those of the European Union, ERC, EIC or SFI. Neither the European Union nor the granting authority can be held responsible for them.}

\acknowledgements{We thank Doan Dai Nguyen, Constantine Evans, Dave Doty, Sergiu Ivanov and Cai Wood for stimulating discussions.}

\nolinenumbers 

\EventEditors{}
\EventNoEds{0}
\EventLongTitle{31st International Conference on DNA Computing and Molecular Programming (DNA31)}
\EventShortTitle{DNA31}
\EventAcronym{DNA}
\EventYear{2025}
\EventDate{August 25--28, 2025}
\EventLocation{École normale supérieure de Lyon, Lyon, France}
\EventLogo{}
\SeriesVolume{}
\ArticleNo{}

\usepackage{algorithm}
\usepackage{algpseudocode}
\usepackage{arydshln}
\usepackage{xcolor}

\usepackage{tikz}
\usetikzlibrary{arrows.meta, bending, positioning, calc}
\definecolor{darkergreen}{rgb}{0.0, 0.65, 0.0}

\definecolor{vlightgray}{gray}{0.98}
\usepackage[textsize=tiny, color=vlightgray 
]{todonotes}

\usepackage[normalem]{ulem}

\usepackage{xspace}

\newcommand{\ssELbase}{SSEL}
\newcommand{\ssEL}{\ensuremath{\#\mathrm{\ssELbase}}\xspace}
\newcommand{\ssELtext}{$\#$\ssELbase\xspace}

\newcommand{\base}[1]{\ensuremath{\mathrm{#1}}\xspace}
\newcommand{\baseA}{\base{A}}
\newcommand{\baseT}{\base{T}}
\newcommand{\baseG}{\base{G}}
\newcommand{\baseC}{\base{C}}
\newcommand{\baseU}{\base{U}}
\newcommand{\PolyM}{PF-polynomially magnifiable\xspace}
\newcommand{\polyM}{PF-polynomially magnifiable\xspace}
\newcommand{\magAdapt}{magnification adaptable\xspace}
\newcommand{\UBP}{BPM\xspace}
\newcommand{\UBPname}{base pair matching\xspace}
\newcommand{\weepara}[1]{{\bf #1.}}

\begin{document}

\maketitle

\begin{abstract}

To understand and engineer biological and artificial nucleic acid systems, 
algorithms are employed for prediction of secondary structures at thermodynamic equilibrium. 
Dynamic programming algorithms are used to compute the most favoured, or Minimum Free Energy (MFE), structure, and the Partition Function (PF)---a tool for assigning a probability to any structure. 
However, in some situations, such as when there are large numbers of strands, or pseudoknoted systems,   NP-hardness results show that such algorithms are unlikely, but only for MFE. 
Curiously, algorithmic hardness results were not shown for PF, 
leaving two open questions on the complexity of 
PF for  multiple strands and single strands with pseudoknots. 
The challenge is that while the MFE problem cares  only about one, or a few structures, PF is a summation over the entire secondary structure space, 
giving theorists the vibe that computing PF should not only be as hard as MFE, but should be even harder.

We  answer both questions. 
First, we show that computing PF is \#P-hard for systems with an unbounded number of strands, answering a question of Condon Hajiaghayi, and Thachuk [DNA27]. 
Second, for even a single strand, but allowing pseudoknots, we find that PF is \#P-hard.
Our proof relies on a novel  {\em magnification trick} 
that leads to a  tightly-woven set of reductions  
between five key thermodynamic problems: 
MFE, PF, their decision versions, and \ssEL that counts structures of a given energy.  
Our reductions show these five problems are fundamentally related for any energy model amenable to magnification. 
That general classification clarifies the mathematical landscape of nucleic acid energy models and yields several open questions.
\end{abstract}

\section{Introduction}
\label{sec:typesetting-summary}
The information encoding abilities of RNA is harnessed by biology to encode myriad complex behaviours, and  
both DNA and RNA have been used by scientists and engineers to create custom nanostructures and molecular computers. 
Consequently, predicting the structures formed by DNA and RNA strands is crucial both for understanding molecular biology and advancing molecular programming.

The  {\em primary structure} of a DNA strand is simply a word over the alphabet  $\{\baseA, \baseC, \baseG, \baseT\}$, with $\baseU$ instead of $\baseT$ for RNA.  
Bases bond in pairs, \baseA-\baseT and \baseC-\baseG, and a set of such pairings for one or more strands is called a {\em secondary structure}. Typically, the set of all secondary structures $\Omega$ has size exponential in the total number of bases. 
Assuming an energy model over secondary structures, each secondary structure $S$ has an associated real-valued free energy $\Delta G(S)$, where more negative means more favourable~\cite{tinoco}.

The Boltzman distribution is used to model the probability distribution of secondary structures at  equilibrium~\cite{mccaskill1990equilibrium, dirks2004algorithm, dirks2007thermodynamic}: the probability of~$S$ is given by  $p(S) = \frac{1}{Z} \mathrm{e}^{- \Delta G(S)/k_\mathrm{B}T}$, where $Z$ is a normalisation factor called the partition function (PF):
\begin{equation}\label{eq:pf}
	Z  = \sum_{S\in\Omega} \mathrm{e}^{- \Delta G(S)/k_\mathrm{B}T} 
\end{equation} 
Intuitively, $Z$ is an exponentially weighted sum of the free energies over $\Omega$, 
where~$k_\mathrm{B}$ is Boltzmann's constant and $T$ is temperature in Kelvin.  
The algorithmic complexity of computing the PF is the main object of study in this paper.

Predicting the free energy of the most favoured secondary structure(s) at thermodynamic equilibrium is called the Minimum Free Energy (MFE) problem~\cite{zukerrna, nussinov1978algorithms, zukeroptimal}.
A third problem of interest, called \ssEL~(or number of Secondary Structures at a specified Energy Level),  asks a more fine-grained question~\cite{demaine2024domain}:   
Given a value $k\in\mathbb{R}$ how many secondary structures have free energy $k$? 
Here, we give a mathematical relation between all of these models (\cref{fig:reduction_map}) that is somewhat energy-model agnostic and settle the computational complexity of PF in two settings (\cref{tab:}).

\setlength\tabcolsep{4 pt}
{\small
\begin{table}[t]
\centering
\begin{tabular}{m{2.6cm}|m{2.9cm}||m{3.2cm}|m{3.8cm} }
\textbf{Structure} & \textbf{Num. Strands} & \textbf{MFE} & \textbf{PF} \\ \midrule
\multirow{3}{2.5cm}[-2mm]{Unpseudoknotted} 
  & 1 
  & P~\cite{zukerrna, nussinov1978algorithms, zukeroptimal} 
  & P~\cite{mccaskill1990equilibrium} 
\\ \cdashline{2-4}
  & Bounded ($\mathcal{O}(1)$) 
  & P~\cite{ShalabyWoods} 
  & P~\cite{dirks2007thermodynamic} 
\\ \cdashline{2-4}
  & Unbounded ($\mathcal{O}(n)$)
  & \textnormal{\mbox{NP-hard~\cite{condon2021predicting}}}; \UBP~model
  & \textbf{\#P-hard} [\cref{Cor:Pf_shHard_unb}]; \UBP model 
\\
\midrule

Pseudoknotted 
  & 1 
  & NP-hard~\cite{akutsu2000dynamic, lyngso2000rna, Lyngso2004}; including  BPS model
  & \textbf{\#P-hard}~[Theorem~\ref{th_PFhardness}]; BPS model
\\
\bottomrule
\end{tabular}
\caption{\label{table}
	Results on the computational complexity of MFE and PF.  
	P: problem is solvable in time polynomial  in $n$, the total number of DNA/RNA bases;
	NP-hard: as hard as any problem in nondeterministic polynomial time (NP);  
	\#P-hard: as hard as counting the accepting paths of an NP Turing machine.
	All positive results showing a problem in P hold for all 3 models studied here: \UBPname (\UBP), base pair stacking (BPS), and  nearest neighbour (NN).  
    Negative, or hardness, results are proven in the simple \UBP or BPS models 
    and are conjectured~\cite{condon2021predicting} to hold in the more complex NN model. 
 Specifically: Condon et al.~\cite{condon2021predicting} use the \UBP model, and Lyngs\o ~\cite{Lyngso2004} uses the BPS model to prove the NP-completeness of the decision version of MFE. 
 The two \#P-hardness\footref{ft:sharpP} results are contributions of this paper, along with the reductions of \cref{fig:reduction_map}.  
}\vspace{-3ex}
\label{tab:}
\end{table}
}

\begin{figure}[t]
	\includegraphics [width=1\textwidth]{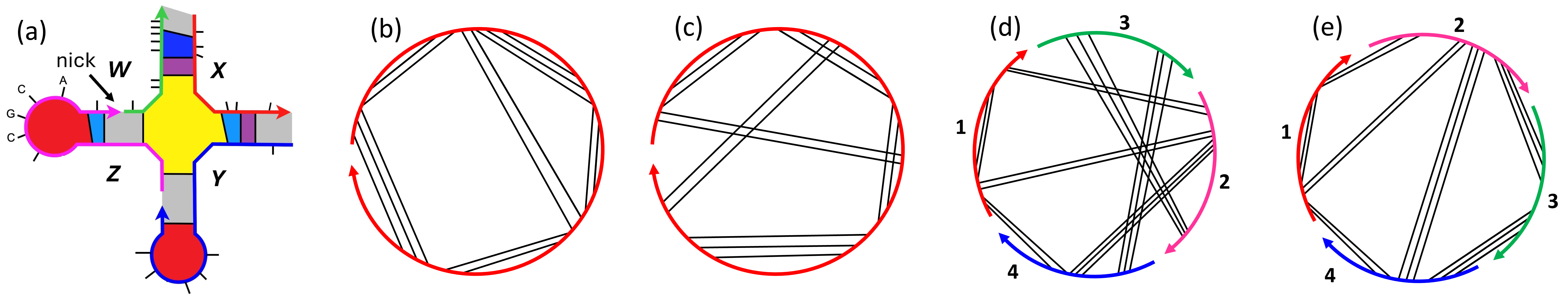}\vspace{-2ex}
	\caption{The nearest neighbour, or Turner, model of DNA/RNA multistranded secondary structure. 
		(a)~One of the many possible secondary structures for four DNA (or RNA) strands $W,X,Y,Z$. 
		Short black lines represent DNA bases (a few are shown $\ldots \baseC, \baseG, \baseC, \baseA \ldots$), and long lines represent base pairs (not to scale). 
		Loops are colour-coded:  stack (purple), hairpin (red), bulge (light blue), internal (dark blue), multiloop (yellow),  external (grey).    
		Black arrow: {\em nick}, the gap between two strands.  
	(b--c) Polymer graphs\footnotemark\xspace for two single-stranded secondary structures:  (b) unpseudoknotted (no crossings), (c) pseudoknotted (crossings). 
	(d) Polymer graph for a secondary structure $S$ over the strand set $\{1, 2, 3, 4\}$ with strand ordering $1324$ showing crossings.
	(e) Simply by reordering the strands to $1234$ gives a polymer graph without crossings, proving $S$ is unpseudoknotted.}
\label{fig:polymer}
\end{figure}
\footnotetext{A secondary structure can be represented as a polymer graph by ordering and depicting the directional ($5'$ to $3'$) strands around the circumference of a circle, with edges along the circumference representing adjacent bases, and straight line edges connecting paired
bases. Each such ordering of $c$ strands is a circular permutation of the strands, and there are $(c-1)!$ possible orderings~\cite{condon2021predicting,ShalabyWoods,dirks2007thermodynamic}.}

\subsection{Background and related work} \label{sec:related}
\weepara{Efficient algorithms}
Decades ago, the relationship between secondary structure prediction and dynamic programming algorithms was well established. 
For a single strand of length $n$, dynamic programming techniques were used to solve the MFE and PF problems efficiently, meaning in polynomial time in the number of bases, but ignoring so-called pseudoknotted structures. 
Early algorithms were designed for simple energy models that count the number of base pairs, the \UBPname (\UBP) model~\cite{zukerrna, nussinov1978algorithms}. Later algorithms accounted for more complex secondary structure features, 
including stacks, hairpin loops, internal loops and other features that comprise the {\em nearest nieghbour}, or {\em Turner}, model~\cite{mccaskill1990equilibrium,dirks2007thermodynamic}, \cref{fig:polymer}, extensively used by practitioners~\cite{NUPACK, viennaRNA, mfold}.

\weepara{Hardness results for single strands and an open problem}
Somewhat frustratingly, dynamic programming algorithms for MFE prediction do not handle all secondary structures: 
as noted, pseudoknots (\cref{def_pk}) throw a spanner in the works. 
Pseudoknotted structures are ubiquitous in both biological RNA and DNA nanotech and computing systems, so why ignore them?  
In 2000, prediction  in the presence of pseudoknots was shown to be NP-hard, even for a single strand and under simple energy models like the base pair stacking model (BPS; counts stacks)~\cite{akutsu2000dynamic,lyngso2000rna,Lyngso2004}.
NP-hardness implies we are unlikely to see efficient algorithms for the full class of pseudoknots~\cite{garey1979computers}, although  progress has been made on specific subclasses~\cite{dirks2004algorithm, rivas1999dynamic, akutsu2000dynamic,dirks2003partition,chen2009n,jabbari2018knotty}. Although it feels as though computing PF should be at least as hard as MFE, for reasons outlined below the complexity of PF in this setting (single strand, allowing pseudoknots) remained tantalisingly open.

\weepara{Hardness results for multiple strands and another open problem}
Recently, Condon, Hajiaghayi, and Thachuk~\cite{condon2021predicting} proved that computing 
MFE is also NP-hard in the unpseudoknotted scenario in the \UBP model when the number of DNA/RNA strands is unbounded, 
meaning it scales with problem size.
They left the complexity of PF in this setting as an open problem.

\weepara{Results on \ssELtext}
The counting problem \ssEL~is also known in the literature as the density of states problem~\cite{DensityOfStates}.
Efficient dynamic programming algorithms can be adapted to solve a version of this problem where energies fall into discrete bins (ranges), for unpseudoknotted secondary structures~\cite{Cupal1996Density}. 
In addition, probabilistic heuristic methods have been proposed to estimate its value with good accuracy~\cite{DensityOfStates}.
These heuristics originate from statistical physics, particularly from the study of the Ising model. 
Interestingly, this model draws a useful parallel, as computing its partition function is known to be either in \textsf{P} or \#\textsf{P}-complete, depending on the parameters~\cite{goldberg2009complexitydichotomypartitionfunctions}.

\weepara{Indirect evidence that PF might be harder than MFE}
In 2007, Dirks et al.~\cite{dirks2007thermodynamic} gave a polynomial time dynamic programming algorithm for PF for multiple, but $\mathcal{O}(1)$, strands in the NN model~\cite{dirks2007thermodynamic}. 
Their paper includes a definitional contribution that extends the single-strand NN model to multiple strands by including both a strand association penalty and an entropic penalty for rotational symmetry of multi-stranded secondary structures. 
They observe if the model is permitted to ignore rotational symmetry, any purely
dynamic programming algorithm for PF can be translated into an algorithm for MFE using tropical $(\min,+)$ algebra instead of the classic $(+,\cdot)$ algebra. 
But going the other way is not so obvious: translating MFE algorithms into PF algorithms is challenging due to the risk of  overcounting secondary structures, 
by translating to an overly naive  algorithm. 
This observation provides some intuition that PF might be harder than MFE.

Recently Shalaby and Woods~\cite{ShalabyWoods} gave an efficient algorithm for computing MFE in the same setting as Dirks et al.~\cite{dirks2007thermodynamic} ($\mathcal{O}(1)$ strands, unpseudoknotted). 
For roughly two decades, this setting had supported an efficient algorithm for PF but none for MFE due to rotational symmetry complications. 
The rareness of the situation 
and its positive resolution 
bolstered our intuition PF is not easier than MFE.

\weepara{A note about energy models} 
Efficient algorithms seem to work across a range of energy models (\UBP, BPS, NN),\footnote{The BPS model is a special case of the NN model, so any algorithm in the NN model can be modified for the BPS model by easily   ignoring all loops except stacks and setting $\Delta G^{\mathrm{stack}}= -1$.} 
but NP-hardness results have merely been shown for MFE in simple energy models: \UBP and BPS. 
As Condon, Hajiaghayi, and Thachuk observe~\cite{condon2021predicting}, 
it seems unlikely MFE would become easier in more sophisticated models like NN.

\begin{figure}[t]
	\centering
	\vspace{-4ex}
    \begin{tikzpicture}[
        scale=1, transform shape,
		node distance=4cm and 4cm,
		every node/.style={circle, draw, fill=gray!10, minimum size=1.8cm, text centered, font=\sffamily},
		arrow/.style={-{Stealth[scale=1.5]}, thick},
		parallel/.style={{Stealth[scale=1.5]}-{Stealth[scale=1.5]}, thick, double distance=2mm},
		assumption/.style={font=\bfseries\sffamily}
		]
		\node (MFE) at (90:3cm) {MFE};
		\node (dMFE) at (18:3cm) {dMFE};
		\node (dPF) at (306:3cm) {dPF};
		\node (PF) at (234:3cm) {PF};
		\node (EL) at (162:3cm) {\ssEL};
		
		\draw[arrow, gray] (MFE.-30) -- (dMFE.140) node[draw=none, fill=none, xshift=-12mm, yshift=3mm, font=\small]{Th. \ref{th:easy}};
		\draw[arrow, gray] (dMFE.110) -- (MFE.00) node[draw=none, fill=none, xshift=14mm, yshift=-4mm, font=\small]{Th. \ref{th:easy}};
		\draw[arrow, gray] (EL) -- (MFE) node[draw=none, fill=none, xshift=-20mm, yshift=-8mm, font=\small]{Th. \ref{th:easy}};
		\draw[arrow, gray] (EL.-60) -- (PF.90) node[draw=none, fill=none, xshift=2mm, yshift=10mm, font=\small]{Th. \ref{th:easy}};
		\draw[arrow, red] (PF.120) -- (EL.-90) node[draw=none, fill=none, xshift=-3mm, yshift=-10mm, font=\small]{Th. \ref{PFtoEL}};
		\draw[arrow, red] (dPF) -- (dMFE) node[draw=none, fill=none, xshift=-13mm, yshift=-17mm, font=\small]{Th. \ref{dMFE_to-dPF}};
		\draw[arrow, red] (dPF.170) -- (PF.10) node[draw=none, fill=none, xshift=10mm, yshift=3mm, font=\small]{Th. \ref{dPFtoPF}};
		\draw[arrow, gray] (PF.-20) -- (dPF.200) node[draw=none, fill=none, xshift=-10mm, yshift=-3mm, font=\small]{Th. \ref{th:easy}};
		
		\draw[arrow, black] (dMFE) -- (EL) coordinate[pos=0.5] (crossing);
		\node[draw=none, fill=none, font=\sffamily, scale=3] at (crossing) {\textcolor{black}{X}};
		\node[draw=none, fill=none, xshift=0mm, yshift=-6.5mm, font=\small] at (crossing){Th. \ref{th_no-exi}};
	\end{tikzpicture}\vspace{-4ex}
	\caption{Reduction map between the five main problems (\cref{subsec:prob_def}) studied in this paper, results shown in \cref{sec:reductions}. 
	An arrow from problem A to B means that if there exists an algorithm for A, it can be called to efficiently solve B, or in other words that there is a polynomial-time Turing reduction from B to A.
	Red arrows use a magnification of the energy model, grey arrows do not. 
	A crossed arrow signifies that no such reduction exists unless $\text{\#P} \subseteq \text{P}^{\text{NP}}$, which means the collapse of polynomial hierarchy~\cite{arora2009computational} at level 2. The latter implies the non-existence of arrows from dMFE to dPF and from MFE to \ssEL.}\vspace{-3ex}
	\label{fig:reduction_map}
\end{figure}
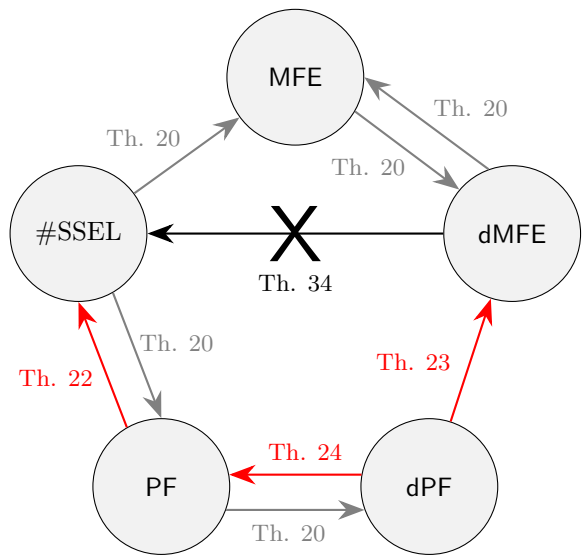

\subsection{Contributions}

We prove hardness results on the complexity of PF in two settings ({\cref{tab:}}), 
and provide tools for thinking about the complexity of thermodynamic prediction problems  ({\cref{fig:reduction_map}}). 

\cref{sec:BG} contains definitions of secondary structures, energy models (\UBP, BPS and NN), and thermodynamic problems (MFE, PF, \ssEL, and decision problems dMFE and dPF.) 

Our first main contribution is to provide a map of reductions between all of these problems, illustrated in \cref{fig:reduction_map} with proofs in \cref{sec:reductions}. 
Some of these reductions are rather straight-forward (grey arrows in~\cref{fig:reduction_map}), but the others (red arrows) make use of a new proof strategy we call the energy {\em magnification trick}.
In \cref{subsec:red_discussion} 
we define a property of an energy model called {\em \polyM} (\cref{def:mag})
which means that the energy model has a polynomial time \magAdapt PF algorithm. That general property yields \cref{cor:polyM puts problem in P}, that all 5 problems in \cref{fig:reduction_map} are in P (or FP) whenever the underlying energy model is  \polyM.
\cref{subsec:red_discussion} has a full discussion about \polyM energy models and \magAdapt algorithms.

Our second main contribution is to answer two open problems on the complexity of ~PF: 
we show that PF is \#P-hard\footnote{A first upperbound of all five problems in \cref{fig:reduction_map} is the exponential time class EXP, as we can solve any problem just by going though all distinct secondary structures. Therefore, the interesting question is about finding better complexity lowerbounds and upperbounds. The complexity class \#P, introduced by Leslie Valiant~\cite{Valiant79}, is the class of problems where the goal is to count the number of solutions, with each solution having a polynomial-sized certificate.\label{ft:sharpP}} in (a) the single-stranded case, i.e.~with pseudoknots under the BPM model, and (b) in the multi-stranded case even without pseudoknots (\cref{tab:}) under the BPS model.
These results are proven in \cref{sec:P-hardness_BPS,sec:P-hardness_Condon}, and leverage the reduction map (\mbox{\cref{fig:reduction_map}}, \mbox{\cref{sec:reductions}}). 
Since dMFE is NP-complete~\cite{condon2021predicting,Lyngso2004},
these \#P-hardness results provide a strong result, showing that PF is strictly harder than dMFE in these two situations, unless $\text{\#P} \subseteq \text{P}^{\text{NP}}$, which implies the collapse of the polynomial hierarchy~\cite{arora2009computational} at level 2.

In addition to showing the 
\#P-hardness lowerbound on PF, the reductions in \cref{sec:reductions} show that PF and \ssEL have equivalent complexity in the sense that they are polynomial-time Turing reducible~\cite{arora2009computational} to each other, providing an upperbound on the complexity of~PF. 

Although our $\text{\#P}$-hardness results are shown for the BPM and BPS models, our reductions apply to the NN model, and \cref{App:candlev} has some analysis of the NN model.

\subsection{Future work}
\begin{enumerate}
\item 
 MFE hardness in two settings is still an open question: Is MFE NP-hard for 
NN single-strand pseudoknotted, or NN $\mathcal{O}(n)$-strands unpseudoknotted? If so, can the reduction map be leveraged to show hardness results for PF and \ssEL in the NN model?

\item Using the reduction map (\cref{fig:reduction_map}) for positive results: 
Is the NN model \polyM (see \cref{def:mag})? 
(The existing algorithm by Dirks et al.~\cite{dirks2007thermodynamic} does not handle rotational symmetry in a \magAdapt way.)
A positive answer,  and the reduction map would immediately imply a polynomial time algorithm for \ssEL. 
(There already is a polynomial time algorithm for MFE~\cite{ShalabyWoods}, and hence dMFE.) 

\item \cref{fig:reduction_map} has arrows of two colours: 
red that means the reduction uses our magnification trick, and grey does not. 
Can all red arrows be made grey? 
This involves replacing our magnification trick with a completely different strategy. 

\item In \cref{rem:poly},  we assume 
in the NN model that loops are specified using at most logarithmic precision. 
Logarithmic precision seems reasonable from a physical perspective, since parameters to the NN free energy model have uncertainty after only a few  decimal places. 
Mathematically, some loop free energies, namely, hairpin, interior and bulge loops, are a logarithmic function of their size and thus may be irrational. Hence we ask: Can this logarithmic precision assumption be dropped?

\item We studied the MFE problem, which aims to compute the minimum free energy value. A natural variant is to ask for its corresponding secondary structures. How does this version relate to our five other problems?
\end{enumerate}

\section{Definitions}\label{sec:BG}
First, we need to review some basic terminology to formulate our main problems in an easy and precise way. We set the scene, in \cref{sinmuldef}, by the mathematical definitions for single-stranded nucleic acid systems before extending them to the multi-stranded case.
In \cref{energymodels}, we provide a brief review of existing energy models and some abstract properties of these models that play a significant role in our reductions.
Finally, in \cref{subsec:prob_def}, we provide the formal definitions of the main problems of interest in this work.

\subsection{Single-stranded and multi-stranded nucleic acid systems} \label{sinmuldef}

Formally, a DNA strand $s$ is a word over the alphabet of DNA {\em bases} $\{\mathrm{A},\mathrm{T},\mathrm{G},\mathrm{C}\}$, indexed from 1 to $|s|=n$, where $n$ denotes the number of bases of $s$.
Hydrogen bonds, or base pairs, can form between complementary bases, namely C–G and A–T, and formally such a base pair is written as a tuple $(i, j)$, where $i<j$, of strand indices (indexing from 1).

\begin{definition2}[Single-stranded secondary structure $S$]\label{def:single strand sec struct}
    For any DNA strand $s$, a secondary structure $S$ of $s$ is a set of base pairs, such that each base appears in at most one pair, i.e. if $(i, j) \in S$ and $(k, l) \in S$ then $i, j, k, l$ are all distinct or $(i, j)= (k, l)$.
\end{definition2}

\begin{definition2}[Unpseudoknotted single-stranded secondary structure]\label{def_pk}
    A secondary structure $S$ of a strand $s$ is unpseudoknotted if for any two base pairs $(i,j)$ and $(k,l) \in S$, $i<k<j$ if and only if $i<l<j$. Otherwise, we call $S$ pseudoknotted, see \cref{fig:polymer}.
\end{definition2}

\begin{remark}\label{rk_check_ss}
	The maximum number of base pairs in any  secondary structure $S$ of strand $s$ is $\lfloor n/2 \rfloor$. Hence,
	$S$ is representable in space  polynomial in $n$, and  verifying whether $S$ is a valid secondary structure is computable in polynomial time.
\end{remark}

We extend these definitions to a multiset s of $c \in \{1,2,3,\ldots\}$ interacting nucleic acid strands, with 
$c=1$ is called \textbf{single-stranded} and $c>1$ \textbf{multi-stranded}.
Each strand has a unique identifier  $t \in  \{1,\ldots,c\}$~\cite{dirks2007thermodynamic}, and 
a base $i_t$ has a strand identifier $t$ and an index $1 \leq i \leq l_t$, where $l_t$ is the length of the $t^\text{th}$ strand.  
In this case, $n$ is the total number of bases: $n=\sum_{1\leq t \leq c} l_t$.

\begin{definition2}[Multi-stranded secondary structure $S$]\label{def:multi strand sec struct}
   For $c\in\mathbb{N}$ interacting strands, a secondary structure $S$ is a set of base pairs such that each base appears in at most one pair, i.e.~if $(i_n, j_m)\in S$ and $(k_q, l_r)\in S$ then $i_n,j_m,k_q,l_r$ are all distinct or $(i_n, j_m)= (k_q, l_r)$.
\end{definition2}

\begin{definition2}[Unpseudoknotted multi-stranded secondary structure]
   For  $c\in\mathbb{N}$ interacting  strands, we call a secondary structure $S$ unpseudoknotted if there exists at least one ordering $\pi$ of the $c$ strands such that if we consider the $c$ strands in a row as forming one single long strand, with lexicographically ordered bases, then $S$ is unpseudoknotted according to \cref{def_pk}, see \cref{fig:polymer}.
\end{definition2}

\subsection{Energy models}\label{energymodels}
We are interested in the computational complexity and relationships between five problems (\cref{subsec:prob_def}) with an energy model as part of their input. 
First, we review three important energy models from the literature. 
\cref{App:candlev} provides an analysis of the set of {\em  candidate energy levels} for each model.

\begin{definition2}[Energy model]\label{def:abstract energy model}
	An \textbf{energy model}$\mathcal{M}$ defines a free energy function $\Delta G^\mathcal{M}$, such that $\Delta G^\mathcal{M}: \Omega_s \times \mathbb{R}^+ \rightarrow \mathbb{R}$, assigns a real value $\Delta G^\mathcal{M}(S,T)$ to any secondary structure $S$ of strand $s$, given a temperature $T$ (in Kelvin), where $\Omega_s$ is the set of all secondary structures (under interest) of $s$.
\end{definition2}

\begin{note}[Magnification of an energy model] \label{note:magnification}
For a positive real number $\alpha\in\mathbb{R}^+$ the notation $\alpha \cdot \Delta G^\mathcal{M}$ simply means to multiply the free energy function $\Delta G^\mathcal{M}(S,T)$ by $\alpha$. Whenever we apply magnification in this paper, we do so uniformly over all secondary structures $S$, which does not change their relative free energy ordering nor distribution per free energy level. 
This magnification is simple to compute for any given $S$ (just a multiplication), however it may not be obvious how to modify a sophisticated PF or MFE algorithm to be magnification compatible, as discussed in \cref{subsec:red_discussion}.
\end{note}

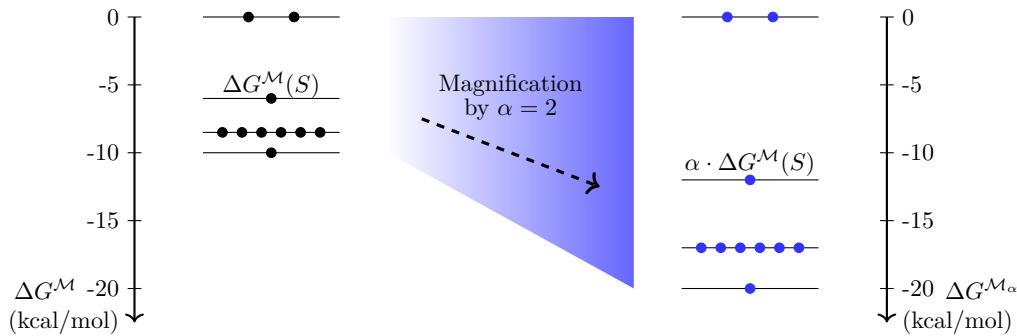
\begin{figure}[h]
\centering
\begin{tikzpicture}
    \tikzset{
    scale=0.9, transform shape
    }
    \draw[->, thick] (0,0) -- (0,-4.5);

    \foreach \y/\label in {0/0, -1/-5, -2/-10, -3/-15, -4/-20} {
        \draw (-0.1,\y) -- (0.1,\y); 
        \node[left] at (-0.1,\y) {\label}; 
    }
    \node[left] at (-0.1,-4.5) {(kcal/mol)};
    \foreach \y in {0,-2,-1.2,-1.7} {
    \draw (1,\y) -- (3,\y); 
    }

    \newcommand{\placeNodes}[3]{
        \pgfmathsetmacro{\N}{#1}
        \pgfmathsetmacro{\a}{1}
        \pgfmathsetmacro{\b}{3}
        \pgfmathsetmacro{\step}{(\b - \a)/(\N + 1)}
        \foreach \i in {1,...,\numexpr\N} {
            \pgfmathsetmacro{\x}{\a + \i * \step}
            \filldraw[black] (\x, #2) circle (2pt);
            }
    }

    \placeNodes{2}{0};   
    \placeNodes{6}{-1.7};  
    \filldraw[black] (2, -2) circle (2pt);
    \filldraw[black] (2, -1.2) circle (2pt);
    \draw (2,-1) node {$\Delta G^\mathcal{M}(S)$};
    \draw (9,-2.15) node {$\alpha \cdot \Delta G^\mathcal{M}(S)$};

    \draw (-1.3,-4) node {$\Delta G^\mathcal{M}$};
    \draw (12.4,-4) node {$\Delta G^\mathcal{M_\alpha}$};

    \shade[left color=white, right color=blue!60] (3.7,0) -- (7.3,0) -- (7.3,-4) -- (3.7,-2) -- cycle;

    \draw[->,dashed, line width=1.2pt] (4.2,-1.5) -- (6.8,-2.5)
    node[above,pos=0.5, yshift= 2em] {Magnification}
    node[above,pos=0.5, yshift=1em] {by $\alpha = 2$}; 

    \draw[->, thick] (11,0) -- (11,-4.5);

    \foreach \y/\label in {0/0, -1/-5, -2/-10, -3/-15, -4/-20} {
        \draw (11.1,\y) -- (10.9,\y); 
        \node[right] at (11.1,\y) {\label}; 
    }
    \node[right] at (11.1,-4.5) {(kcal/mol)};
    
    \foreach \y in {0,-2,-1.2,-1.7} {
    \pgfmathsetmacro{\newy}{2*\y}
    \draw (8,\newy) -- (10,\newy); 
    }

    \newcommand{\placeNodesright}[3]{
        \pgfmathsetmacro{\N}{#1}
        \pgfmathsetmacro{\a}{8}
        \pgfmathsetmacro{\b}{10}
        \pgfmathsetmacro{\step}{(\b - \a)/(\N + 1)}
        \foreach \i in {1,...,\numexpr\N} {
            \pgfmathsetmacro{\x}{\a + \i * \step}
            \filldraw[blue!80] (\x, #2) circle (2pt);
            }
    }

    \placeNodesright{2}{0};   
    \placeNodesright{6}{-1.7*2};  
    \filldraw[blue!80] (9, -2*2) circle (2pt);
    \filldraw[blue!80] (9, -1.2*2) circle (2pt);

\end{tikzpicture}
\caption{Illustration of the magnification process used in \cref{PFtoEL} and \cref{dMFE_to-dPF}. 
Each node is the free energy of a secondary structure $S$, with nodes on the left being $\Delta G^\mathcal{M}(S)$ and the right being $\Delta G^{\mathcal{M}_\alpha}(S)=\alpha\cdot \Delta G^\mathcal{M}(S)$.
Magnification increases the absolute value of all energy levels, without changing the distribution of secondary structures per energy level.}
\end{figure}

\begin{definition2}[Base pair matching (\UBP) model]
	\label{def:UBP}
	In the \UBP model, the free energy of any secondary structure $S$, denoted by $\Delta G^\mathrm{\UBP}(S)$, is the number of base pairs formed in $S$, such that each is weighted $-1$, hence: $ \Delta G^\mathrm{\UBP}(S)= -|S|$. 
\end{definition2}

Despite the simplicity of the \UBP model~\cite{nussinov1980fast}, it is still powerful enough to prove hardness results.
For example, Condon, Hajiaghayi, and Thachuk~\cite{condon2021predicting} used it to prove the NP-hardness of MFE prediction of an unbounded set of strands in the unpseudoknotted case.
In 2004, Lyngs\o~\cite{Lyngso2004} introduced another energy model, called the base pair stacking (BPS) model, 
and proved  NP-hardness of MFE  for single-stranded pseudoknotted systems with stacks.

\begin{definition2}[Base pair stacking (BPS) model] \label{def:BPS}
	The number of base pair stackings $(\mathrm{BPS})$ of any secondary structure $S$ is defined as
	$\mathrm{BPS}(S)=|\{(i,j)\in S~|~(i+1,j-1)\in S\}|$.
	In the BPS model, the free energy of a secondary structure $S$ is $\Delta G^\mathrm{BPS}(S)= -\mathrm{BPS}(S)$.
\end{definition2}

\begin{note} \label{note:temp}
	We say that an energy model $\mathcal{M}$ is \textbf{temperature independent} if its energy function 
	$\Delta G^\mathcal{M}$ is not a function of  temperature. Hence, $\Delta G^\mathcal{M} (S,T) = \Delta G^\mathcal{M} (S)$ for any $T$.
	The \UBP and BPS models are temperature independent, and the NN model (below) is not.
\end{note}

A word of caution.  
 From classical thermodynamics~\cite{tinoco,dirks2007thermodynamic}: 
 free energies are typically of the form $\Delta G=\Delta H-T\Delta S$ meaning they are a function of temperature $T$ and can be decomposed into enthalpic ($\Delta H$) and entropic ($\Delta S$) contributions. 
Whether DNA/RNA binding occurs is strongly temperature dependent, which is why free energies are too. 
Thus, temperature independent energy models are not  good representations of 
typical temperature-varying scenarios, however
they are a useful vehicle to show computational complexity results in a fixed-temperature setting. 

Beyond temperature dependence,
the \UBP and BPS models do not account for a number of features of DNA and RNA that 
provide significant free energy contributions: single-stranded regions and global symmetry.
This motivates our third, most realistic, energy model which is called the nearest neighbour (NN) model:

\newcommand{\NNDG}{\ensuremath{\Delta G^\mathrm{NN}(S) = \sum_{l\in \mathrm{loops}(S,s)} \Delta G(l) + (c-1)\Delta G^{\textrm{assoc}}  + k_\mathrm{B} T \log R}}

\begin{definition2}[Nearest neighbour (NN) model]\label{def:NN}
	Let $S$ be an unpseudoknotted connected\footnote{In the NN model, unlike the \UBP and BPS models, 
	multi-stranded secondary structures must be connected~\cite{dirks2007thermodynamic}, 
	meaning the polymer graph of a secondary structure is connected.} 
	secondary structure over a multiset $s$ of $c\geq 1$ strands. 
	$S$ s decomposed into
	{a multiset of loops, denoted $\mathrm{loops}(S,s)$, each being one of the types}: 
	hairpin, interior, exterior, stack, bulge, and multiloop, as described in \cref{fig:polymer}. 
	Then, the free energy of $S$ is the sum: 
	$$\NNDG$$ 
	  where $\Delta G : \mathrm{loops}(S,s) \rightarrow \mathbb{R}$ 
	  gives a free energy for each loop~\cite{santa},
	$\Delta G^{\mathrm{assoc}} \in \mathbb{R}$ is an association penalty 
	applied $c-1$ times for a complex of $c$ strands~\cite{dirks2007thermodynamic}, 
	and $R$ is the maximum degree of rotational symmetry~\cite{ShalabyWoods,dirks2007thermodynamic} of $S$, details follow. 
\end{definition2}

A few comments are warranted on \cref{def:NN}. 
At fixed temperature, salt concentration, etc., the loop free energy $\Delta G(l)$ for a stack $l$ is simply a function of the stack's arrangement of its four constituent DNA/RNA bases~\cite{santa}. 
For loops $l$ with single-stranded regions (e.g.~interior, hairpin) $\Delta G(l)$ is a logarithmic function of loop length~\cite{dirks2007thermodynamic,santa}, although for multiloops a linear approximation is used to facilitate dynamic programming~\cite{dirks2007thermodynamic}.  
We write the multiset of loops as a function of both the secondary structure $S$ and strand $s$ to emphasise that both base pair indices, and base identities are used to define loops.  In this work, we consider $\Delta G^{\mathrm{assoc}}$ to be a constant, however see~\cite{dirks2007thermodynamic} for more details, including temperature and water-molarity dependence.
Versions of the NN model are implemented in software suites like \texttt{NUPACK}~\cite{dirks2007thermodynamic,fornace2020unified}, \texttt{ViennaRNA}~\cite{viennaRNA}, and \texttt{mfold}~\cite{mfold}.  For more analysis of the NN model, see \cref{App:candlev}.

\begin{definition2}[Set of candidate energy levels] \label{def:can}
    Given an energy model $\mathcal{M}$ and strand $s$ (or a set of strands $s$), a \textbf{set of candidate energy levels} $\mathcal{G}_s^\mathcal{M}$ is a finite superset of the energies of all secondary structures of $s$, in other words  $\{ \Delta G^\mathcal{M}(S) \mid S \in \Omega_s  \} \subseteq \mathcal{G}_s^\mathcal{M}$. 
\end{definition2}

\begin{note}\label{rem:poly}

  $\mathcal{G}_s^\mathcal{M}$ is defined as a {\em superset} so that it is easily computable: specifically, computing $\mathcal{G}_s^\mathcal{M}$ does not require computing the MFE. 
    
    In \cref{App:candlev} we show that there are  
    sets of candidate energy levels of polynomial size and computable in polynomial time for all three models studied. 
    For \UBP and BPS the proofs are straightforward, but for NN we rely on an assumption that individual loop free energies are specified using logarithmic precision, by which we mean they can be written down as a rational number using $O(\log n)$ digits in units of kcal/mol. 
Physically, this is a reasonable assumption since measuring such free energies beyond a few decimal places is likely quite challenging. 
Mathematically, the assumption is not a trivial one, 
since hairpin, interior and bulge loops involve taking logarithms of natural numbers resulting in irrational free energies. 
However, computationally the assumption seems reasonable as any implemented algorithm will have finite precision.
\end{note}

\input{problems}

\input{reductions}

\input{bps}

\input{mbps}

\bibliography{bib}

\input{omittedproofsApp}

 \end{document}

%% file: problems.tex

\subsection{Definitions of problems: MFE, PF, dMFE, dPF and \ssELtext}\label{subsec:prob_def}

In this section, we define the main five problems for which we establish computational complexity relationships.
For convenience, we present the definitions for a single strand---the multi-stranded case is defined similarly, but the strand $s$ is replaced by a multiset of   strands, and $n$ denotes the sum of strand lengths, or total number of bases, of the system. 


\begin{definition2}[MFE; a function problem]~
	\begin{description}
		\item[Input:]  Nucleic acid strand $s$ of length $n$, a temperature $T \geq 0$, and an energy model $\mathcal{M}$.
		\item[Output:] The minimum free energy $\mathrm{MFE}(s,T, \mathcal{M}) = \min\{\Delta G^\mathcal{M}(S,T) \mid S \in \Omega_s\}$, where $\Omega_s$ is the set of all secondary structures (under interest) of $s$.
	\end{description}
\end{definition2}

\begin{definition2}[PF; a function problem]~
	\begin{description}
		\item[Input:] Nucleic acid strand $s$ of length $n$, a temperature $T \geq 0$, and an energy model $\mathcal{M}$.
		\item[Output:] The partition function $\mathrm{PF}(s,T,\mathcal{M}) = \sum_{S \in \Omega_s} e^{-{\Delta G^\mathcal{M}(S,T)}/{k_{\mathrm{B}} T}}$, where $\Omega_s$ is the set of all secondary structures (under interest) of $s$.
	\end{description}
\end{definition2}

\begin{definition2}[dMFE; a decision problem]~
	\begin{description}
		\item[Input:] Nucleic acid strand $s$ of length $n$, a temperature $T \geq 0$, an energy model $\mathcal{M}$ and a value $k\in \mathbb{R}$.
		\item[Output:] Is $\mathrm{MFE}(s,T,\mathcal{M}) \leq k$?
	\end{description}
\end{definition2}

\begin{definition2}[dPF; a decision problem]~
	\begin{description}
		\item[Input:] Nucleic acid strand $s$ of length $n$, a temperature $T \geq 0$, an energy model $\mathcal{M}$, and a value $k\in \mathbb{R}$.
		\item[Output:] Is $\mathrm{PF}(s,T,\mathcal{M}) \geq k$?
	\end{description}
\end{definition2}

\begin{definition2}[\ssELtext; a counting problem]~
	\begin{description}
		\item[Input:] Nucleic acid strand $s$ of length $n$, a temperature $T \geq 0$, an energy model $\mathcal{M}$, and a value $k\in \mathbb{R}$.
		\item[Output:] $\ssEL(s,T,\mathcal{M}, k)$: the number of secondary structures $S$ of $s$ such that $\Delta G^\mathcal{M}(S)=~k$.
	\end{description}
\end{definition2}

\
\begin{definition2}[Polynomial-time Turing reduction]
    A polynomial-time Turing reduction~\cite{arora2009computational} from a problem $A$ to a problem $B$ is an algorithm that solves problem A using a polynomial number of calls to a subroutine for problem B, and polynomial time outside of those subroutine calls.
\end{definition2}

%% file: reductions.tex
\section{Reductions between the computational thermodynamic problems}\label{sec:reductions}
 \cref{fig:reduction_map} illustrates most of the results of this section.
For convenience, all proofs are written for the single-stranded case. 
However, all results hold in the multi-stranded case, simply by replacing the unique input strand by a set of multiple strands in the proofs.

\subsection{Straightforward reductions} \label{subsec:lemma_and_easy_proofs}
In this subsection, we prove all straightforward reductions (gray colored) in our reduction map in \cref{fig:reduction_map}. 
In the last three reductions, we use our assumption in \cref{rem:poly} about the set of candidate energy levels.

\begin{theorem}\label{th:easy}
There exist the following  polynomial-time Turing reductions: $\mathrm{dMFE}$ to $\mathrm{MFE}$, $\mathrm{dPF}$ to $\mathrm{PF}$, $\mathrm{MFE}$ to $\mathrm{dMFE}$, $\mathrm{MFE}$ to $\ssEL$, and $\mathrm{PF}$ to $\ssEL$.
\end{theorem}

\begin{proof} 

\begin{enumerate}
	\item \textbf{Reduction from dMFE to MFE:} Let $(s,T,\mathcal{M},k)$ be an input for the dMFE problem. 
	After a single call to the $\text{MFE}(s,T,\mathcal{M})$ function, simply computing the Boolean value $(\text{MFE}(s,T,\mathcal{M}) \leq k)$ gives the answer to the dMFE problem. 
	\label{MFEtodMFE}
	
	\item \textbf{Reduction from dPF to PF:} Let $(s,T,\mathcal{M},k)$ be an input for the dPF problem. Similarly, if the partition function $\mathrm{PF}(s,T,\mathcal{M})$ is known, then the Boolean value $(\mathrm{PF}(s,T,\mathcal{M}) \geq k)$ is the answer to the dMFE problem.\label{PFtodPF}
	
	\item \label{dMFEtoMFE}\textbf{Turing Reduction from MFE to dMFE:} Let $(s,T,\mathcal{M})$ be an input for the MFE problem. We know that $\mathrm{MFE}(s, T, \mathcal{M}) = g_j$ for some $1 \leq j \leq |\mathcal{G}_s^\mathcal{M}|$ by definition of $\mathcal{G}_s^\mathcal{M}$. 
	Determining the MFE is equivalent to determining the integer $j$, which can be achieved through a binary search over $\mathcal{G}_s^\mathcal{M}$ thanks to the dMFE oracle.
	The complexity of this search is $\mathcal{O}(\log(|\mathcal{G}_s^\mathcal{M}|)) \in \mathcal{O}(\mathrm{poly}(n))$, as $\mathcal{G}_s^\mathcal{M}$ is of polynomial size.
	This binary search defines a Turing reduction from MFE to dMFE.

	\item \label{ELtoMFE}\textbf{Turing Reduction from MFE to \ssEL}: Let $(s,T, \mathcal{M})$ be an input for the MFE problem, to determine $\mathrm{MFE}(s, T, \mathcal{M})$, we linearly search for $j = \max_{1 \leq i \leq |\mathcal{G}_s^\mathcal{M}|} \{i \mid \ssEL(s, T, \mathcal{M}, g_i) \neq 0\}$ (i.e. the maximal such $i$ gives the MFE due to how the indices are ordered). This search requires only a polynomial number of calls to the \ssEL  oracle, giving a polynomial time Turing reduction from MFE to \ssEL.
	
	\item \label{ELtoPF}\textbf{Turing Reduction from PF to \ssEL:} Let $(s,T,\mathcal{M})$ be an input for the PF problem. The partition function PF$(s, T, \mathcal{M}) = \sum_{1 \leq i \leq |\mathcal{G}_s^\mathcal{M}|} \ssEL(s, T, \mathcal{M}, g_i) e^{-g_i / k_\text{B}T}$. This computation is computable in polynomial time as it requires a polynomial number of calls to the \ssEL oracle, hence defining a Turing reduction from PF to \ssEL.    
\end{enumerate}
\end{proof}

\subsection{Polynomial-time Turing Reduction from \ssELtext to PF}\label{subsec:red_EL_to_PF}
In this section, we prove one of the red-arrow reductions 
in \cref{fig:reduction_map}.
The number of secondary structures of a strand $s$, denoted  $\# \mathrm{SecStruct}(s)$, plays an important role in our reductions. Note that $\# \mathrm{SecStruct}(s)$ includes pseudoknotted and unpseudoknotted structures.

 \begin{restatable}{lemma}{upperbound}
	\label{lemma:upboun}
	For any strand $s$ of size $n > 2$,  $\# \mathrm{SecStruct}(s) < n!$.
 \end{restatable}

\begin{proof}
	
	By enumerating all secondary structures based on the number of base pairs formed in each. Any secondary structure $S$ has $k \leq \lfloor n/2 \rfloor$ base pairs, these $k$ base pairs are formed between $2k$ different bases. The first base, among the $2k$ different bases, has $2k-1$ bases to form a base pair with. The next “unpaired” base has $(2k-3)$ bases to form a base pair with. At the end, the number of secondary structures which have exactly $k$ base pairs is upper bounded by $\binom{n}{2k}\cdot(2k-1)\cdot(2k-3)\cdots3\cdot1$. Leading to     
	\begin{align}
		\# \mathrm{SecStruct}(s) &\leq \sum_{k=0}^{\lfloor n/2 \rfloor}\binom{n}{2k}\cdot(2k-1)\cdot(2k-3)\cdots \label{eq:ub}&\\
		&\leq \sum_{k=0}^{\lfloor n/2 \rfloor}\frac{n!}{(2k)! \ (n-2k)!} \ \frac{1 \cdot 2 \cdot 3 \cdots 2k}{2 \cdot 4 \cdot 6 \cdots 2k} &~\\
		&\leq \sum_{k=0}^{\lfloor n/2 \rfloor}\frac{n!}{(2k)! \ (n-2k)!} \ \frac{(2k)!}{k! \ 2^k}&~ \\
		&\leq n! \sum_{k=0}^{\lfloor n/2 \rfloor} \frac{1}{(n-2k)! \  k! \ 2^k}&~\\ 
		& \leq n! \sum_{k=0}^{\lfloor n/2 \rfloor} \frac{1}{\lfloor n/3 \rfloor! \ 2^k}  & \text{as }\sum_{k=0}^{\lfloor n/2 \rfloor} \frac{1}{2^k} < 2\\  
		& < \frac{2n!}{\lfloor n/3 \rfloor!}  \leq n! & \text{for } n\geq 6
	\end{align}
	
	Note that, $\#\mathrm{SecStruct}(s) = n!$ for $n=1$ or $n=2$, which are useless cases for any system. To extend the inequality to all $n > 2$, we only need to check if it is verified for $3\leq n\leq 5$. By computing the right-hand side of \cref{eq:ub}, the associated numbers of possible secondary structures are at most $4,10,\text{and }26$ which are upper bounded by its corresponding $n!$. Therefore,  $\# \mathrm{SecStruct}(s)  < n!$ for all $n>2$.  
\end{proof}

\begin{theorem} \label{PFtoEL}
	There exists a polynomial-time Turing reduction from \ssEL~to $\mathrm{PF}$.
\end{theorem}
\begin{proof}
For notational convenience, instead of $\ssEL(s,T,\mathcal{M},g_i)$ we write $\ssEL(g_i)$ to denote the number of secondary structures with free energy $g_i$, and $\beta = 1/k_\text{B}T$. 
Let $N= | \mathcal{G}_s^\mathcal{M} |$ denote the size of the set of candidate energy levels, see \cref{def:can,rem:poly}.

\cref{algo1} gives a polynomial time Turing reduction from  \ssEL~to PF, for the inputs $s, T, \mathcal{M}, g_k$, where the main idea behind this algorithm is the following:
\begin{itemize}
	\item We use our \textit{magnification trick} to magnify the distances between energy levels:
	this idea imagines another energy model $\mathcal{M}_j$, 
	where $\Delta G^{\mathcal{M}_j} = j  \cdot \Delta G^\mathcal{M}$ 
	(\cref{def:abstract energy model} defines $\Delta G^\mathcal{M}$)
	which we assume a PF oracle can handle (see \cref{subsec:red_discussion} for details).

	\item We compute PF using a call to the oracle $\mathrm{PF}(s,T,\mathcal{M}_j)$ with energy model $\mathcal{M}_j$.
	
	\item We do this magnification $N$ times with different magnification factors $1 \leq j \leq N$.
	
	\item We end up with a system of linear equations that has a unique solution which is the number of secondary structures 
	per energy levels, outputting the correct one $\ssEL(g_k)$.

\end{itemize}

\begin{algorithm} 
	\caption{Computing \ssEL  by calling the oracle $\mathrm{PF}(s,T,\mathcal{M})$ for PF.} \label{algo1}
	\begin{algorithmic}[1]
		\Statex \textbf{Input:} 
		$s,T,\mathcal{M},g_k$
		\State Compute a set of candidate energy levels $\mathcal{G}_s^\mathcal{M} = \{g_1, \cdots, g_N\}$ \Comment see \cref{App:candlev}
		\For{$j \gets 1$ to $N$}

		\State Let $b_j = \mathrm{PF}(s,T,\mathcal{M}_j)$, where $ \Delta G^ {\mathcal{M}_j}= j \cdot \Delta G^ \mathcal{M}$ \Comment{see \protect{\cref{note:magnification}}} \label{line:3}
		\EndFor
		\State Solve the following system of linear equations with $N$ unknowns $\{\ssEL(g_i), ~1\leq i \leq N\}$ and with $\beta = 1/ k_\mathrm{B} T$: 
		\[
		\begin{array}{c@{\hskip 4pt}c@{\hskip 4pt}c@{\hskip 4pt}c@{\hskip 4pt}c@{\hskip 4pt}c@{\hskip 4pt}c@{\hskip 4pt}c@{\hskip 4pt}c}
        \ssEL(g_1)e^{-\beta g_1} &+& \ssEL(g_2)e^{-\beta g_2} &+& \cdots &+& \ssEL(g_N)e^{-\beta g_N} &=& b_1 \\
        \ssEL(g_1)(e^{-\beta g_1})^2 &+& \ssEL(g_2)(e^{-\beta g_2})^2 &+& \cdots &+& \ssEL(g_N)(e^{-\beta g_N})^2 &=& b_2 \\
        \vdots && \vdots && \ddots && \vdots && \vdots \\
        \ssEL(g_1)(e^{-\beta g_1})^N &+& \ssEL(g_2)(e^{-\beta g_2})^N &+& \cdots &+& \ssEL(g_N)(e^{-\beta g_N})^N &=& b_N
        \end{array}
        \]
    \State \textbf{Return:}  $\ssEL(g_k)$ 
	\end{algorithmic}
\end{algorithm}

\noindent\textbf{Correctness}:
An algorithm to compute the set $\mathcal{G}_s^\mathcal{M} $ is given in \cref{App:candlev}. First note that the partition function can be partitioned  by energy levels $g_i$ to give the form: $\mathrm{PF}(s,T,\mathcal{M})= \sum_{i=1}^{N}\ssEL(g_i)e^{-\beta \cdot g_i} $. 
Hence, under magnification $j$ (see \cref{note:magnification}), the general expression for all \( b_j \) is:
\[b_j =\mathrm{PF}(s,T,\mathcal{M}_j)= \sum_{i=1}^{N}\ssEL(g_i)e^{-\beta\cdot j \cdot g_i} =\sum_{i=1}^{N}\ssEL(g_i)(e^{-\beta\cdot g_i})^{j}\]
which is the form in \cref{algo1}.  \cref{algo1} solves this system of linear equations; to see this we write it  
in the standard matrix form $\mathbf{A}\mathbf{x} = \mathbf{b} $, where: 
\[
\mathbf{A} = \begin{pmatrix}
	e^{-\beta \cdot g_1} & e^{-\beta \cdot g_2} & \cdots & e^{-\beta \cdot g_N} \\
	(e^{-\beta \cdot g_1})^2 & (e^{-\beta \cdot g_2})^2 & \cdots & (e^{-\beta \cdot g_N})^2 \\
	\vdots & \vdots & \ddots & \vdots \\
	(e^{-\beta \cdot g_1})^N & (e^{-\beta \cdot g_2})^N & \cdots & (e^{-\beta \cdot g_N})^N
\end{pmatrix}, \quad
\mathbf{x} = \begin{pmatrix}
	\ssEL(g_1) \\
	\ssEL(g_2) \\
	\vdots \\
	\ssEL(g_N)
\end{pmatrix}, \quad
\mathbf{b} = \begin{pmatrix}
	b_1 \\
	b_2 \\
	\vdots \\
	b_N
\end{pmatrix}.
\]
Since $\mathbf{A}$ is a Vandermonde matrix~\cite{VanDerMonde}, and $e^{-\beta \cdot g_i} \neq e^{-\beta \cdot g_j}$ if $i \neq j$, then matrix $\mathbf{A}$ is non-singular. Therefore, the system of equations $ \mathbf{A}\mathbf{x} = \mathbf{b} $ has a unique solution. 
That solution is the list $[\ssEL(g_1), \ssEL(g_k), \ldots, \ssEL(g_N)]$, the $k$th entry being $g_k$, which is returned. \newline

\noindent\textbf{Time analysis of \cref{algo1}}: 
By the results in \cref{sec:brute force EL,sec:more precise EL algo}, a 
candidate set $\mathcal{G}_s^\mathcal{M} $ is computed in polynomial time.
$N=\mathrm{poly}(n)$, therefore, {\cref{algo1}} calls the PF oracle $\mathrm{PF}(s,T,\mathcal{M}_j)$ a polynomial number of times. Moreover, each  call is has polynomial input size. 
Solving the resulting system of $N$ linear equations can be achieved in $\mathrm{poly}(N)$ time by Gaussian elimination, since the inputs of the system are stored in polynomial space: 
\begin{itemize}
	\item Matrix  $\mathbf{A}$: the largest entry of  $\mathbf{A}$ has size $ \log ((e^{-\beta \cdot g_N})^N) = N \log (e^{-\beta \cdot g_N}) = \mathrm{poly}(n)$, since $g_i = \mathrm{poly}(N) = \mathrm{poly}(n)$, \cref{rem:poly}.
	
    \item Vector $\mathbf{b}$: the largest entry of  $\mathbf{b}$ has size 
    \begin{align*}
        \log (b_N) &= \log(\ssEL(g_1)(e^{-\beta g_1})^N + \cdots + \ssEL(g_N)(e^{-\beta g_N})^N) \\
        & \leq \log(N\mathrm{\#SecStruct(s)}(e^{-\beta g_N})^N) \\
        & \leq \log(N) + \log(\mathrm{\#SecStruct(s)}) + N\log(e^{-\beta g_N}) \\
        & \leq \mathrm{poly}(n) \hspace{4cm} \text{\hfill As $\#\mathrm{SecStruct}(s) < n!$ by \cref{lemma:upboun}.}
    \end{align*}
\end{itemize}
This shows the existence of  a polynomial-time Turing reduction from \ssEL~to PF.
\end{proof}

\subsection{Reduction from dMFE to dPF}\label{subsec:red_dMFE_to_dPF}

\begin{theorem} \label{dMFE_to-dPF}
There exists a polynomial-time Turing reduction from $\mathrm{dMFE}$ to $\mathrm{dPF}$.
\end{theorem}

\begin{proof}

We denote the minimal step between any two energy levels by $\delta$, \cref{rem:poly}, and $\beta = 1/(k_\text{B}T)$. \cref{algo2} gives a polynomial time Turing reduction from  dMFE to dPF, where the main idea behind this algorithm is the following:
\begin{itemize}
	\item We use our magnification trick to make a \textbf{huge} magnification of energy levels. 
	
	\item This magnification is carefully designed so that the contribution to the PF of exactly one secondary structure, the MFE one, overwhelms the contribution to PF of others. 
	
	\item Our construction guarantees this by handling the worst case scenario:
	\begin{enumerate}
		\item Only one secondary structure is at the MFE level.
		\item $\#\mathrm{SecStruct}(s)$ secondary structures, see \cref{lemma:upboun}, belong to the closest energy level to the MFE level, which is MFE $+ ~\delta$. 
	\end{enumerate}

	\item   This huge separation of the contribution of secondary structures is achieved due to the exponential nature of the PF. 
	That, in turn, solves dMFE problem with a dPF oracle. 
\end{itemize}
\begin{algorithm} 
	\caption{Solve dMFE calling an oracle for dPF.}
	\label{algo2}
	\begin{algorithmic}[1]
		\Statex \textbf{Input:} 
		$s,T,\mathcal{M},k$
        \State Compute a set of candidate energy levels $\mathcal{G}_s^\mathcal{M}$\Comment see \cref{App:candlev}
		\State Let $x = \mathrm{max}\{{g \in \mathcal{G}_s^\mathcal{M} ~|~ g \leq k}\}$ 
		\Comment $x$ is the max value from $\mathcal{G}_s^\mathcal{M}$ that is  $\leq k$. 
		\State Let $\Delta G^{\mathcal{M}'} = (\log(n!)/(\beta \delta)) \cdot \Delta G^\mathcal{M}$
		\State Let $k'=e^{-{\log(n!)}\cdot x/\delta}=(n!)^{- x/\delta}$

\Statex \textbf{Return:} 
	     $\mathrm{dPF}(s,T,\mathcal{M}',k')$

	\end{algorithmic}
\end{algorithm}
\textbf{Correctness}:
We need to prove the following claim: $\mathrm{dMFE}(s,T,\mathcal{M},k) \Leftrightarrow \mathrm{dPF}(s,T, \mathcal{M}', k')$, which is equivalent to $\mathrm{MFE}(s,T,\mathcal{M}) \leq k \Leftrightarrow \mathrm{PF}(s,T, \mathcal{M}') \geq k'$. By definition of $x$, we have $\mathrm{MFE}(s,T,\mathcal{M}) \leq k \Leftrightarrow \mathrm{MFE}(s,T,\mathcal{M}) \leq x$. Therefore, it's equivalent to prove the following: 
\[\mathrm{MFE}(s,T,\mathcal{M}) \leq x \Leftrightarrow \mathrm{PF}(s,T, \mathcal{M}') \geq k'.\]
The partition function, after magnification, is the following: \[\mathrm{PF}(s,T,\mathcal{M}') = \sum_{g\in \mathcal{G}_s^\mathcal{M}}\ssEL(s,T,\mathcal{M},g)e^{-\beta\cdot  g\cdot \ln(n!)/(\beta\delta)} =  \sum_{g\in \mathcal{G}_s^\mathcal{M}}\ssEL(s,T,\mathcal{M},g)(n!)^{-g/\delta}\]

$\Rightarrow$ If $\mathrm{MFE}(s,T,\mathcal{M}) \leq x$, then a MFE secondary structure contributes to the magnified partition function, $\mathrm{PF}(s,T,\mathcal{M}')$, with a coefficient of $(n!)^{- \mathrm{MFE}(s,T,\mathcal{M})/\delta} \geq (n!)^{- x/\delta} = k'$, hence  $\mathrm{PF}(s,T,\mathcal{M}') \geq k'$ as required.

$\Leftarrow$ Conversely, if $\mathrm{MFE}(s,T,\mathcal{M}) > x$, then by definition of $\delta$, $\mathrm{MFE}(s,T,\mathcal{M}) \geq x+\delta$.
\begin{align*}
	\mathrm{PF}(s,T, \mathcal{M'}) &= \sum_{g\in \mathcal{G}_s^\mathcal{M}}\ssEL(s,T,\mathcal{M},g)(n!)^{-g/\delta} \\
	&\leq \#\mathrm{SecStruct}(s) \cdot (n!)^{-\mathrm{MFE}(s,T,\mathcal{M})/\delta} \\
	&< n! \cdot (n!)^{-\mathrm{MFE}(s,T,\mathcal{M})/\delta} \hspace{1.5cm} \text{As $\#\mathrm{SecStruct}(s) < n!$ by \cref{lemma:upboun}.}\\
	&< (n!)^{-(\mathrm{MFE}(s,T,\mathcal{M})-\delta)/\delta} \\
	&< (n!)^{-x/\delta} = k' 
\end{align*}
This proves the main claim that: $\mathrm{dMFE}(s,T,\mathcal{B},k) \Leftrightarrow \mathrm{dPF}(s,T, \mathcal{B}', k')$.\newline

\noindent\textbf{Complexity analysis}:
\begin{itemize}
	\item Step 1: is done in polynomial-time, since $|\mathcal{G}_s^\mathcal{M}| = \mathrm{poly}(n)$, \cref{rem:poly}. 
	
	\item Step 2 (magnification step): is done in polynomial-time, since  $\log(\frac{\log(n!)}{\beta\cdot\delta}) \leq  \log(n^2) + \log(k_{\text{B}}T) + \log(1/\delta) \leq \mathrm{poly}(n) + \mathrm{poly}(T)$,  due to the logarithmic size of $1/\delta$, \cref{rem:poly}.
	
	\item Step 3: $k'$ can be computed in $\mathrm{poly}(n)$ time, and the size of $k'$ is $\log((n!)^{- x/\delta}) \leq (-x/\delta)\cdot n^2 = \mathrm{poly}(n)$ bits.
\end{itemize}
This shows the existence of a polynomial-time Turing reduction from dMFE to dPF.
\end{proof}

\subsection{Polynomial-time Turing Reduction from PF to dPF}\label{subsec:red_PF_to_dPF}
This next reduction also exploits our magnification trick, hence is slightly more involved than the previous one from MFE to dMFE.  
while for PF, the search space is of exponential size. A first approach would be to binary search for the PF value, using the oracle dPF. This approach runs in polynomial time if we are satisfied with linear precision of the PF value (but the PF can have arbitrary precision since it uses exponentials of $\mathrm{e}$). 
However, we can search for the exact value by combining the simple binary search with our magnification trick,  efficiently exploiting the dPF oracle to compute the  exact PF contribution at each energy level.

\begin{theorem}\label{dPFtoPF}
 	There exists a polynomial-time Turing reduction from
    $\mathrm{PF}$ to  $\mathrm{dPF}$.
\end{theorem}

\begin{proof}
	Let $(s,T,\mathcal{M})$ be the input of the PF problem, and consider \cref{algo2} and notations used in  \cref{dMFE_to-dPF}. 
	Since $\ssEL(s,T,\mathcal{M},g) < n!$ for all $g \in \mathcal{G}_s^\mathcal{M}$ by \cref{lemma:upboun}, 
	the partition function PF$(s,T, \mathcal{M}') = \sum_{g \in \mathcal{G}_s^\mathcal{M}} \ssEL(s,T,\mathcal{M},g)(n!)^{-g/\delta}$ can be seen as a number written in base $n!$. Determining this number in base $n!$ is equivalent to determining all coefficients of the form $\ssEL(s,T,\mathcal{M},g)$.
	To achieve this, we do the following: 
    \begin{itemize}
    	\item We perform a binary search for each coefficient $\ssEL(s,T,\mathcal{M},g)$ in decreasing order over $\mathcal{G}_s^\mathcal{M}$ {(from the most to the less favourable)}.
    	
        \item First, guess $\ssEL(s,T,\mathcal{M},g_{|\mathcal{G}_s^\mathcal{M}|})$, by calling dPF$(s, T, \mathcal{M}', c_1 (n!)^{- g_{|\mathcal{G}_s^\mathcal{M}|} / \delta})$, {repeatedly, using binary search with search variable $c_1$ over the range} $0 \leq c_1 \leq n!$.
        \item Inductively, after computing  $\ssEL(s,T,\mathcal{M},g_i)$, find $\ssEL(s,T,\mathcal{M},g_{i-1})$ by calling dPF oracle with input $(s, T, \mathcal{M}', \sum_{j=1}^i \ssEL(s,T,\mathcal{M},g_j)(n!)^{- g_j / \delta} + c_{i-1} (n!)^{- g_{i-1} / \delta})$, where $c_{i-1}$ is the search variable: $0 \leq c_{i-1} \leq n!$.
        \item At the end, we know all of the $\ssEL(s,T,\mathcal{M},g_i)$ values, we then combine them to directly compute the partition function $\mathrm{PF}(s,T,\mathcal{M})= \sum_{i=1}^{N}\ssEL(s,T,\mathcal{M},g_i)e^{- g_i/(k_\text{B}T)}$.
    \end{itemize}
	By using binary search, the algorithm requires only $|\mathcal{G}_s^\mathcal{M}| \log(n!)$ calls to dPF oracles,  {which is} $ \text{poly}(n)$, \cref{rem:poly}. Hence, it {is} a polynomial-time Turing reduction from PF to dPF.
\end{proof}

\subsection{Discussion for \cref{sec:reductions}: reductions and  the magnification trick}\label{subsec:red_discussion}

When we call a specific oracle using our \textit{magnification trick}, 
we ask the oracle to function under the same energy model $\mathcal{M}$ but magnified (e.g.~by factor $j$ in Line \ref{line:3} of \cref{algo1}, to give $\mathcal{M}_j$). 
We believe this  kind of magnification of the energy model, combined with our reductions in \cref{sec:reductions}, is a useful notion for finding  polynomial time algorithms for  thermodynamic problems (such as those in \cref{fig:reduction_map}), hence we propose a general definition: 

\begin{definition2}[\PolyM energy model]\label{def:mag}
	An energy model $\mathcal{M}$ is \textbf{\polyM} if 
	there is a polynomial time algorithm $\mathrm{PF}(\cdot,\cdot,\mathcal{M}_j)$ that computes PF under the $j$-magnified energy model: $ \Delta G^ {\mathcal{M}_j}= j \cdot \Delta G^ \mathcal{M}$, for all $j \in \mathrm{poly}(n)$, and we call  $\mathrm{PF}(\cdot,\cdot,\mathcal{M})$ a \textbf{\magAdapt} algorithm under $\mathcal{M}$.  
\end{definition2}

\subsubsection{Positive (polynomial time) results}

In the unpseudoknotted 1-strand, or even $\mathcal{O}(1)$-strand cases, with the \UBP and BPS models, 
there is a  polynomial time algorithm for PF. 
It is not hard to show that PF for \UBP and BPS  is \magAdapt, 
hence \UBP and BPS are \polyM energy models.  
Hence,  from our reductions and that PF algorithm.  
we get polynomial time algorithms for free for the other four problems besides PF in \cref{fig:reduction_map}. 
The same holds for the NN model {\em without rotational symmetry} 
(i.e.~single-stranded NN model, or the multi-stranded NN model that ignores rotational symmetry).

\begin{corollary}\label{cor:polyM puts problem in P}
When $\mathcal{M}$ is a \polyM model, all five problems in \cref{fig:reduction_map} are in P.
\end{corollary}

In the NN model with multiple strands and unpseudoknotted secondary structures,  
Dirks et al.~\cite{dirks2007thermodynamic} gave a polynomial time algorithm for PF, 
and Shalaby and Woods~\cite{ShalabyWoods} gave one for MFE. 
One might ask if Dirks et al.~\cite{dirks2007thermodynamic}, plus the  reductions in \cref{fig:reduction_map} are sufficient to yield a polynomial time algorithm for MFE, but this is not known since we don't know whether the Dirks et al. PF algorithm is magnification adaptable. 
The issue is in how Dirks et al.\ handle rotational symmetry, not by pure dynamic programming, but by computing a naive PF that: 
(1) Completely ignores rotational symmetry, and 
(2) Overcounts \textit{indistinguishable} secondary structures (less symmetry, means more overcounting, see~\cite{dirks2007thermodynamic} for  details).
Then they use an algebraic argument to bring back rotational symmetry simultaneously with canceling out the overcounting effect, an argument that may not be preserved under our magnification trick. 
Hence we do not currently know whether the NN model in its full generality (including rotational symmetry) is \polyM.

\subsubsection{Temperature independent energy models}

If the energy model is not \polyM, but is temperature independent (\cref{note:temp}),  
we can achieve our magnification trick  
simply by instead reducing (i.e.~inverse magnification) the partition function temperature as follows (this assumes, as usual, that the partition function {\em is} temperature dependent):
\enlargethispage{2\baselineskip}
\begin{itemize}
\item Assume we have a polynomial time algorithm  $\mathrm{PF}(\cdot,\cdot,\mathcal{M})$, but want to compute $\text{PF}(\cdot,\cdot,\mathcal{M}')$, where $\Delta G^{\mathcal{M}'} = \alpha  \Delta G^\mathcal{M}$.
\item Call $\text{PF}(s,T',\mathcal{M})$, with $T'=T/\alpha$.
 \begin{align*}
	\text{PF}(s,T',\mathcal{M}) 
	&= \sum_{S\in\Omega} \mathrm{e}^{- \Delta G^\mathcal{M}(S)/k_\mathrm{B}T'} \\
	&= \sum_{S\in\Omega} \mathrm{e}^{- \alpha\Delta G^\mathcal{M}(S)/k_\mathrm{B}T} &\text{$\mathcal{M}$ is temperature independent.}\\ 
	&=\text{PF}(s,T,\mathcal{M}')
\end{align*}

\item Hence, $\mathrm{PF}(\cdot,\cdot,\mathcal{M})$ is \magAdapt. 
\end{itemize}

 \begin{note}[NP-hardness results]
As we mentioned before, the dMFE problem is  NP-hard in: (1) The BPS model when we allow pseudoknots~\cite{Lyngso2004}, and (2) The \UBP model when we have unbounded number of strands~\cite{condon2021predicting} without pseudoknots. Then, \ssEL, PF, and dPF are NP-hard as well in these energy models using our reduction map, \cref{fig:reduction_map}. However, we provide \textbf{sharper} complexity results in the following sections.
 \end{note}

%% file: bps.tex
\section{\#P-hardness of PF and \ssELtext  in  BPS model}\label{sec:P-hardness_BPS}

In this section, we prove  \#P-hardness for the problems $\mathrm{PF}(\cdot,\cdot,\mathrm{BPS})$ and $\mathrm{\ssEL}(\cdot,\cdot,\mathrm{BPS})$ (see \cref{def:BPS} for the BPS model).
Our strategy is to construct a reduction chain from \#3DM to \#4-PARTITION, then from \#4-PARTITION to \#BPS, where \#BPS is the counting problem of the number of secondary structures having exactly $K$ stacks.

First, we state the definitions of the  three problems within the reduction chain, before showing that there exist weakly parsimonious reductions\footnote{A reduction is weakly parsimonious if there is a suitably-computable  relation between the number of solutions of the two problems; depending only on the initial problem instance~\cite{arora2009computational}.} within the chain.

\begin{definition2}[\#3DM]~
	\begin{description}
		\item[Input:] Three finite sets $X$, $Y$, and $Z$ of the same size, and a subset $T \subseteq X \times Y \times Z$.
		\item[Output:] The number of perfect 3-dimensional matchings $M$, where $M \subseteq T$ is a perfect 3-dimensional matching if any element in $X \cup Y \cup Z$ belongs to exactly one triple in $M$.
	\end{description}
\end{definition2}

\begin{definition2}[\#4-PARTITION]~
    \begin{description}
        \item[Input:] A set of \( k \) elements \( \mathcal{A}=\{a_1, a_2, \dots, a_k\} \), a weight function \( w: \mathcal{A} \to \mathbb{Z}^+ \), and a bound \( B \in \mathbb{Z}^+ \) such that the weight of each element \( w(a_i) \) is strictly between \( B/5 \) and \( B/3 \). 
        \item[Output:] The number of partitions of $\mathcal{A}$ into 4-tuples, such that the sum of the weights of the elements in each tuple is equal to \( B \). 
    \end{description} 
\end{definition2}

\begin{definition2}[\#BPS]~
	\begin{description}
		\item[Input:] A nucleic acid strand $s$ and a number $K \in \mathbb{N}$.
		\item[Output:] The number of secondary structures $S$ of $s$ that have $K$ base pair stackings.
	\end{description}
\end{definition2}

We prove that \#BPS is \#P-complete through two consecutive weakly parsimonious reductions:
from \#3DM, to \#4-PARTITION, and from \#4-PARTITION to \#BPS. 
\begin{restatable}{lemma}{partitionshp}
	\label{le:4PA_shP}
	$\text{\#}4\text{-PARTITION}$ is \#P-complete.
\end{restatable}

For the sake of completeness, we first include a quick overview of the 4-PARTITION NP-completeness from Garey and Johnson (1979) \cite{GareyJohnson}, as we heavily depend on it in our reduction. After that we give the proof of \cref{le:4PA_shP}.

\begin{proof}[Proof overview of the NP-completeness of 4-PARTITION]~\newline 
	Given an instance of 3DM with three disjoint sets $X$, $Y$, $Z$, with $|X|=|Y|=|Z|$, and a set \(T\subseteq X\times Y\times Z\), the reduction constructs an instance $(\mathcal{A},w,B)$ of 4-PARTITION as follows:
	\begin{itemize}
		\item For each element $x \in X$, we introduce $N(x)$ elements $(x[i])_{1\leq i \leq N(x)}$ in $\mathcal{A}$, where $N(x)$ is the number of times $x$ appears in $T$. These elements have two possible weights, depending on whether $i=1$ or not. The element $x[1]$ is called \textit{actual} while the other ones are called \textit{dummy}. Same happens to every $y \in Y$, and $z \in Z$.
		
		\item  For each triple $(x,y,z) \in T$, an element $u_{xyz}$, is introduced in $\mathcal{A}$, whose weight depends on $x$, $y$ and $z$.
		
		\item The bound $B$ and the weight function $w$ are constructed such that the following equivalence hold:
		A 4-tuples has weight equaling $B$ iff it is of the form $(u_{xyz},x[i],y[j],z[k])$, with $i,j,k$ being all equal to 1, or all greater than 1. In that case, we call the 4-tuple either \textit{dummy} or \textit{actual}. 
		
	\end{itemize}
	Given a 3D-matching $M$, we can construct the following 4-partition $P$. For each triple $(x,y,z) \in M$ , we add to $P$ the \textit{actual }4-tuple $(u_{xyz}, x[1], y[1],z[1])$. For each triple $(x,y,z) \in T \setminus M$, we add to $P$ the \textit{dummy }4-tuple $(u_{xyz}, x[i], y[j],z[k])$, with $i,j$,and $k$ all greater than 1. It is possible to construct these triples since there are enough \textit{dummy} elements. $P$ is a solution of 4-PARTITION since it is a partition (all elements are part of a 4-tuple) and since the weight of each 4-tuple equals $B$ by construction.
	
	Reciprocally, given a 4-partition $P$ of $\mathcal{A}$, we construct a 3D-matching $M$. All 4-tuples of $P$ have weight equaling $B$ and are of the expected form.
	We introduce in $M$ the triples $(x,y,z) \in T$ for the ones $u_{xyz}$ is part of an \textit{actual} 4-tuple in $P$. $M$ is a matching since there is exactly one \textit{actual} element $a[i]$ per element $a\in X\cup Y\cup Z$. Moreover, $M$ is perfect since every \textit{actual} element belongs to an \textit{actual} 4-tuple.
\end{proof}
%
\begin{proof} 
	Now, we prove that in the mentioned reduction from 3DM to 4-PARTITION \cite{garey1979computers}, the number of solutions (4-tuples) equals the number of perfect matchings in the initial instance of 3DM multiplied by a coefficient that depends only on that initial instance.
	
	Let $M$ and $M'$ two distinct perfect 3D-matchings.
	There exists $(x, y, z) \in M \setminus M'$. Therefore, $(u_{xyz}, x[1], y[1], z[1])$ belongs to a 4-partition that corresponds to $M$, while the 4-tuple $(u_{xyz},x[i],y[j],z[k])$ in a 4-partition corresponding to $M'$, must have $i$, $j$, and $k$ being greater than 1.
	Therefore, the two 4-partitions corresponding to $M$ and $M'$ are distinct.

	Reciprocally, let $P$ and $P'$ are two 4-partitions corresponds to the same 3D-matching.
	Each tuple of $P$ and $P'$ is of the form $(u_{xyz},x[i], y[j], z[k])$ where $(x,y,z) \in T$ and $i,j,k$ are equal to 1 or greater than 1.
	The Garey and Johnson's proof states that these two partitions have the same collections of \textit{actual} 4-tuples.
	Therefore, they differ in their \textit{dummy} 4-tuples (with $i,j,k$ greater than 1), which means that some \textit{dummy} elements happens in different order. 
	
	For each element $a \in X \cup Y\cup Z$, there are $N(a) - 1$ \textit{dummy} elements. Therefore, there are $ \alpha = \prod_{a \in X \cup Y \cup Z} (N(a) - 1)!$ different ways to arrange these elements in dummy collections.
	We can conclude that there are $\alpha$ distinct 4-partitions corresponding to the same 3D-matching.
	
	This implies that the number of solutions of 4-PARTITION equals the number of perfect matchings in the initial instance of 3DM multiplied by multiplied by $\alpha$, hence this reduction is weakly parsimonious. 
	
	It is straighforward that \#4-PARTITION belongs to \#P as 4-PARTITION belongs to NP.
	Since \#3DM is \#P-hard according to Bosboom et al. \cite{Bosboom2020}, then \#4-PARTITION is \#P-complete.
\end{proof}

\begin{theorem}\label{th:BPS_shP}
	\#BPS is \#P-complete.
\end{theorem}

\begin{proof}    
    In 2004, Lyngs\o~\cite{Lyngso2004} proved that the decision problem BPS is NP-complete, with a reduction from BIN-PACKING.
	We adapt his reduction from 4-PARTITION to BPS.
	Let $(\mathcal{A} =\{a_1, \cdots, a_k\},w, B)$ be an instance of 4-PARTITION.
    Let us first note that in any instance of 4-PARTITION, $k$ is a multiple of 4, and $B k/4 = \sum_{i=1}^kw(a_i)$, otherwise, it is straightforward to see that this instance has no solution.
    We proceed exactly as Lyngs\o, constructing the following DNA strand $s$: 
	\[
	s = C^{w(a_1)}AC^{w(a_2)}A\cdots AC^{w(a_k)}AAA\underbrace{G^BAG^BA\cdots AG^B}_{k/4\text{ substrings of } G's}
	\]
	and setting the target $K = \sum_{i=1}^k w(a_i) - k$.
	
	As A bases can only form base pairs with T bases, all base pairs in a secondary structure of s will be C-G base
    pairs.
    If a substring $C^{w(a_i)}$ binds with at least two disctinct substrings $G^B$, then it accounts for at most $w(a_i) - 2$ in the BPS score. Since the other substrings $C^{w(a_j)}$ account for at most $w(a_j) - 1$, then $\mathrm{BPS}(S) \leq K-1$.
    Hence, we can find a structure S with $\mathrm{BPS}(S) = K$ iff we can partition the $k$ substrings $C^{w(a_i)}$ into $k/4$ groups that can each be fully base paired using one substring $G^B$; i.e. the total length of the substrings of C’s in any group can be at most $B$. 
    
    It means that $\mathrm{BPS}(S) = K$ iff we can partition the $k$ elements $a_i$ into $k/4$ groups that each has total weight $\leq B$. Since $B = \frac{4}{k}\sum_{i=1}^kw(a_i)$ and since the weight of each element in $\mathcal{A}$ is strictly between $B/5$ and $B/3$, each group has a total weight of exactly $B$ and contains exactly 4 elements.
    Therefore, $\mathrm{BPS}(S) = K$ iff 4-PARTITION problem has a solution.

	Consider two distinct 4-partitions $P_1$ and $P_2$. Let $(a_i, a_j, a_l, a_m)$ belong to $P_1$ but not to $P_2$. The secondary structure corresponding to $P_1$ has the group  ($C^{w(a_i)}$, $C^{w(a_j)}$, $C^{w(a_l)}$, $C^{w(a_m)}$) fully bounded with a substring $G^B$. The secondary structure corresponding to $P_2$ does not have these bindings, as $(a_i, a_j, a_l, a_m)$ does not belong to $P_2$. Therefore, all distinct 4-partitions map to distinct secondary structures.
	
	Conversely, let $S_1$ and $S_2$ two secondary structures of $s$, having $K$ base pair stackings and corresponding to the same 4-partition. Consequently, the groups of 4 elements bounded with the substrings $G^B$ are the same in $S_1$ and $S_2$, up to permutations.
    Specifically, it means that $(k/4)!~(4!)^{k/4}$ secondary structures of $s$ having $K$ BPS map to the same 4-PARTITION solution. The coefficient $(k/4)!$ accounts for the permutation over the substrings $G^B$, while the coefficient $4!$ accounts for the permutation between the substrings $C^{w(a_i)}$ within a same group (i.e.\ paired to a same substring $G^B$).
	
	Therefore, the number of solutions is multiplied by $(k/4)!~(4!)^{k/4}$ and this reduction is weakly parsimonious. 
    Moreover, it is straightforward to see that \#BPS is in \#P, as BPS is in NP.
    Since \#4-PARTITION is \#P-complete by \cref{le:4PA_shP}, \#BPS is \#P-complete also.
\end{proof}

\begin{theorem} \label{th_PFhardness}
	$\mathrm{\ssEL}(\cdot,\cdot,\mathrm{BPS})$ is \#P-complete and $\mathrm{PF}(\cdot,\cdot,\mathrm{BPS})$ is \#P-hard.
\end{theorem}

\begin{proof}
	 $\mathrm{\ssEL}(\cdot,\cdot,\mathrm{BPS})$ belongs to \#P and is equivalent to the problem \#BPS. Hence, $\mathrm{\ssEL}(\cdot,\cdot,\mathrm{BPS})$ is \#P-complete.	
	The reduction map (\cref{fig:reduction_map}),  specifically  \cref{PFtoEL},  implies that computing the partition function is \#P-hard in the BPS model.
\end{proof}

\begin{theorem}\label{th_no-exi}
	In the BPS model, there is no polynomial-time Turing reduction from \ssEL~to dMFE, unless $\#P \subseteq P^{NP}$.
\end{theorem} 
\begin{proof}
	dMFE is NP-complete in the BPS model according to Lyngs\o~\cite{Lyngso2004}, and \ssEL~is \#P-complete according to \cref{th_PFhardness}.
    Suppose that there exists a polynomial-time Turing reduction from \ssEL~to dMFE (in the BPS model): $\text{\ssEL} \in \text{P}^{\text{dMFE}}$.
	
    As these problems are complete, $\text{\ssEL}\in \text{P}^{\text{dMFE}} \Rightarrow \#\text{P} \subseteq \text{P}^{\text{NP}}$.
\end{proof}

\begin{remark}
	If $\#\text{P} \subseteq \text{P}^{\text{NP}}$, then $\text{P}^{\text{\#P} }\subseteq \text{P}^{\text{P}^{\text{NP}}} = \text{P}^{\text{NP}} = \Delta_2^\text{P} \subseteq \Sigma_2^\text{P}$. According to Toda's theorem \cite{Toda}, $\text{PH}\subseteq \text{P}^\text{\#P}$. Therefore, $\text{PH} \subseteq \Sigma_2^\text{P}$: the Polynomial Hierarchy would collapse at level 2.
\end{remark}

%% file: mbps.tex
\section{\#P-hardness of PF and \ssELtext for unpseudoknotted secondary structures of an unbounded set of strands in the \UBP model}\label{sec:P-hardness_Condon}

In this section, we prove \#P-hardness of the two restricted problems $\mathrm{PF}(\cdot,\cdot,\mathrm{\UBP})$ and $\mathrm{\ssEL}(\cdot,\cdot,\mathrm{\UBP})$, see \cref{def:UBP} (\UBP model), in the scenario of unpseudoknotted structures of an unbounded set of strands. Specifically, we show that the problem \#MULTI-PKF-SSP is \#P-complete. The main result of this section is proven in \cref{sec:main_5}, using  lemmas from \cref{sec:requ_lemma_5}.

\begin{definition2}[\#MULTI-PKF-SSP]~
	\begin{description}
		\item[Input:] $c$ nucleic acid strands and a positive integer $k$.
		\item[Output:] The number of unpseudoknotted secondary structures containing exactly $k$ base pairs formed by the $c$ strands.
	\end{description}
\end{definition2}

\subsection{Main results and proofs}
\label{sec:main_5}

\begin{theorem}
    \label{th:MULTI-PKF-SSP}
	\#{\rm MULTI-PKF-SSP} is \#P-complete.
\end{theorem}

\begin{proof}
To prove \#P-completeness, we first consider the reduction provided by Condon, Hajiaghayi, and Thachuk 
for the NP-completeness of the associated decision problem~\cite{condon2021predicting}. 
Using the same notation, we show that this reduction is weakly parsimonious. 
Since this reduction contains many details, we do not recall all of them here.

We first note that the reduction in \cite{condon2021predicting} is from a variation of 3-Dimensional Matching, 3DM(3), in which each element appears in at most 3 triples.
Without loss of generality, this reduction can be extended from another version of 3DM, in which all triples are distinct.

Condon, Hajiaghayi, and Thachuk showed (in Lemma 12 of \cite{condon2021predicting}) that any solution of 3DM, can be transformed into a secondary structure of the new instance $I'$ containing $P$ base pairs, where $P$ is the maximum number of possible base pairs ($P=\min(A,T) + \min(G,C)$). It is easy to see that this transformation from 3DM to MULTI-PKF-SSP is injective.

For the reverse transformation, we show an intermediate result in \cref{le:opt_form}, ensuring the form of any secondary structure of the new instance $I'$ containing $P$ base pairs. This lemma is stated and proven in \cref{sec:requ_lemma_5}.
According to this lemma, the back transformation is clear: considering the perfect triples provides a solution for 3DM.
Let two secondary structures leading to the same 3DM solution. Consequently, they have the same perfect and center-trim-deprived triples.
    Let us make a little dichotomy here:
	\begin{itemize}
		\item If we consider the strands as indistinguishable, then the two secondary structures are the same and the reduction is parsimonious.
		\item Else if we consider the strands as distinguishable, then the two secondary structures differ as some of their trim-complement and separator support strands are permuted. Since there are $2n$ separator strands, $2n$ separator-complement strands and $2m+n$ trim-complement strands, there are $(2n)!^{2}\cdot(2m+n)!$ distinct secondary structures mapping to the same 3DM solution.
	\end{itemize}
	In both cases, there is a simple relation between the number of solutions of the two problems, and hence the reduction is weakly parsimonious.

    \#MULTI-PKF-SSP belongs to \#P, since MULTI-PKF-SSP is in NP.
	Since \#3DM is \#P-complete even in the case of disctinct triples \cite{Bosboom2020}, \#MULTI-PKF-SSP is also \#P-complete.
\end{proof}

\begin{theorem}
    \label{Cor:Pf_shHard_unb}
	$\mathrm{\ssEL}(\cdot,\cdot,\mathrm{\UBP})$ is \#P-complete and $\mathrm{PF}(\cdot,\cdot,\mathrm{\UBP})$ is \#P-hard, in the scenario of unpseudoknotted structures of an unbounded set of strands.
\end{theorem}

\begin{proof}
It is clear that $\mathrm{\ssEL}(\cdot,\cdot,\mathrm{\UBP})$ belongs to \#P and is equivalent to the problem \#MULTI-PKF-SSP in this scenario. Hence, $\mathrm{\ssEL}(\cdot,\cdot,\mathrm{\UBP})$ is \#P-complete.
	
The reduction map (\cref{fig:reduction_map}) and specifically  \cref{PFtoEL} implies that computing the partition function is \#P-hard in this scenario as well.
\end{proof}

\subsection{Lemmas used in the proof of \cref{th:MULTI-PKF-SSP}}
\label{sec:requ_lemma_5}
\begin{lemma}
	Let $S$ be an optimal secondary structure of $I'$ with $P$ base pairs. 
	In $S$, if a trim-complement strand is bound with a center-trim, it cannot also be bound with another trim.
\end{lemma}

\begin{proof}
	By absurd, suppose that $S$ has a trim-complement strand which is bound with a center-trim but also with another trim.
	Since all C's are paired in $S$, two adjacent bases of this trim-complement strand are paired with bases belonging to distinct domains $Trim_1$ and $Trim_2$. Since one of these two domains is a center-trim, they are non-adjacent and are separated by the non-empty sequence $u$. Since there is no pseudoknot, bases of $u$ can only be paired with each other.
	Without loss of generality, we can consider that $Trim_2$ is a middle-trim.
	Therefore, $u$ ends with a flank $x~Sep~y~Sep~z$. 
	Since all A and T bases are paired in $S$ and since there is no pseudoknot, two adjacent bases of this flank must be paired with two adjacent bases of $u$.
	Therefore, $u$ must contain the complementary flank $\overline{z}~\overline{Sep}~\overline{y}~\overline{Sep}~\overline{x}$, which is absurd since all triples are different by hypothesis and since all xyz-support strands and their complement are distinct by construction.
\end{proof}

\begin{lemma}\label{le:center-trim_nb}
	Let $S$ be an optimal secondary structure of $I'$ with $P$ base pairs. 
	$S$ has exactly $n$ perfectly bound center-trim domains. 
\end{lemma}

\begin{proof}
	Let us first show that $S$ has at least $n$ perfectly bound center-trim domains.
	Since there are $2m+n$ trim-complement strands and every C's are paired in $S$, at least $nE$ bases belonging to center-trim domains are paired. According to the previous lemma, if one base of a center-trim is paired, then the entire domain is paired. Therefore, there are at least $n$ perfectly bound center-trim domains.
	
	Let us now suppose that there are $n+i$ bound center-trim domains, with $1 \leq i \leq m-n$. They are perfectly bound according to the previous lemma.
	We call a flank isolated if its neighboring trim-domains are both bound with trim-complementary strands.
	Let us count the number of isolated flanks.
	There are $m-n-i$ center-trim deprived triples. All of them are bound to at most $2(m-n-i)E$ trim-complement bases. Therefore, there are at least $[2m+n-2(m-n-i) - (n+i)]\cdot E = (2n+i)E$ end-trim bases that participate in creating isolated flanks with one of the $n+i$ bound center-trims. Therefore, there are at least $(2n+i)$ isolated flanks.
	
	For each isolated flank, consider the set of strands to which it is bound. 
	This set is not bound to the rest of the template domain since there is no pseudoknot.
	According to Lemma 10 of Condon et al. \cite{condon2021predicting}, the total number of support strands is $10n$. 
	Since there are at least $2n+i$ isolated flanks, at least one of them is bound to less than 5 support strands. We can apply Lemma 17 of Condon et al. \cite{condon2021predicting} to this flank, since the ACT-unpairedness is 0 in $S$. Therefore, this flank must be bound to exactly 5 support strands.
	Recursively, the lemma assumptions still hold, and one can continue applying it until one has 0 available support strand but still $i$ isolated flanks respecting the assumptions, which is absurd.
	Therefore, we know that at most $n$ center-trim domains are perfectly bound.
	
	Finally, there are exactly $n$ perfectly bound center-trim domains. 
\end{proof}

\begin{lemma}\label{le:opt_form}
	Each optimal secondary structure of instance $I'$ having $P$ base pairs is of the following form: $n$ perfect triples, $m-n$ triples whose both end-trims are bound to trim-complement strands, and whose 5' and 3' flanks are paired together.
\end{lemma}
\begin{proof}
	Let $S$ be an optimal secondary structure of $I'$ with $P$ base pairs.
	
	According to \cref{le:center-trim_nb}, $S$ has exactly $n$ fully connected center-trim domains. The remaining trim-complement strands must necessarily be fully bound with trim domains, as there is no alternative configuration for them.
	
	Following the same reasoning as in the previous proof, we recursively apply Lemma 17 from Condon et al. to the isolated flanks. Consequently, all support strands are bound. Since there are no unpaired A or T bases, each xyz-domain is paired with its complement.
	
	To conclude, we note that the 5' and 3' flanks of trim-deprived triples must be bound together. 
\end{proof}

%% file: omittedproofsApp.tex

\appendix

\section{Appendix: Candidate energy levels for three free energy models}\label{App:candlev}
The goal of this section is to show, for an instance of any of the three energy models, \UBP, BPS and NN, that there is a set of candidate energy levels of polynomial size and  polynomial time computable in the number of bases $n$.

\subsection{\UBP and BPS models}
Given a secondary structure, or list of base pairs, $S$ for a strand multiset $s$ with $n$ bases in total, in the \UBP model, the free energy is  $ \Delta G^\mathrm{\UBP}(S)= -|S|$ since $\Delta G^\mathrm{\UBP}$ simply counts and negates the number of base pairs. 
Since $|S| \leq \lfloor n/2 \rfloor$,  \cref{rk_check_ss}, 
the set  $\mathcal{G}_s^\text{\UBP} = \{i \in \mathbb{N}: i \leq \lfloor n/2 \rfloor\}$ is a valid set of candidate energy levels. 
The minimal gap, denoted $\delta$, between any two energy levels is 1.

In the BPS model, the free energy of a secondary structure is simply the negation of the number of stacks: 
 $\Delta G^\mathrm{BPS}(S)= -\mathrm{BPS}(S)$ where $\mathrm{BPS}(S)=|\{(i,j)\in S~|~(i+1,j-1)\in S\}|$.
Hence,
$\mathcal{G}_s^\text{BPS} = \{i \in \mathbb{N}: i \leq \lfloor n/2 \rfloor\}$
is a valid set of candidate energy levels   
and the minimal gap $\delta$ between any two energy levels is 1.

Hence, for both \UBP and BPS, there is a set of candidate levels of polynomial size and polynomial time computable, in $n$.

\subsection{NN model}
In the NN model, a secondary structure $S$ has temperature-dependent free energy (\cref{def:NN}):
$$\NNDG$$
where $\Delta G^{\mathrm{assoc}}$ is the  association~\cite{dirks2007thermodynamic} penalty applied to each of the $c-1$ strands added to the first strand to form a complex of $c$ strands,  $R$ is the maximum degree of rotational  symmetry~\cite{ShalabyWoods,dirks2007thermodynamic} of $S$, $k_\mathrm{B}$ is Boltzmann's constant, and $T$ is temperature in Kelvin.
The set of base pairs $S$ defines structural elements called \textbf{loops} namely: hairpin loops, interior loops, exterior loops, stacks, bulges, and multiloops, as shown in \cref{fig:polymer} and described in the literature~\cite{tinoco,santa,mathews1999expanded}.  
The free energy contribution $\Delta G(l) \in \mathbb{R}$ of each loop $l$ is defined by a temperature-dependent formula parameterised from experimental data~\cite{santa,mathews1999expanded} as described next.

\begin{itemize}
	\item \textbf{Interior loops}:   
    An interior loop~\cite{dirks2003partition, mathews1999expanded} defined by closing pairs $(i,j)$ and $(d,e)$ has sides of length $l_1 = d-i-1$ and $l_2 = j-e-1$, where $l_1 \neq 0$ and $l_2 \neq 0$. The interior loop size $m= l_1 + l_2$, and the loop asymmetry is $\sigma = |l_1-l_2|$. 
		 
    In the standard  NN model, special temperature-dependent energy expressions, with constants (i.e.~constant precision at fixed temperature), are used for the energy contributions of interior loops with either $l_1 \leq 3$ or $l_2 \leq 3$, where there is a constant number of interior loops.\footnote{Note that: stacks and bulges are excluded in this case, as we analyze them separately.}
	For all cases when both $l_1 \geq 4$ and $l_2 \geq 4$, the energy contribution is approximated using the following function~\cite{dirks2003partition, mathews1999expanded}:
	\begin{equation*}\label{eq:internal}
		\Delta G^\text{interior}_{i,d,e,j} = \gamma_1(m) + \gamma_2(\sigma) + \gamma_3(i,j,i+1,j-1) + \gamma_4(e,d,e+1,d-1)
	\end{equation*}
    
	Where $\gamma_1$ is a function of loop size $m$, 
	$\gamma_1$ is a function loop asymmetry $\sigma$, 
	and  $\gamma_3,\gamma_4$ are functions of the 
	identity of the closing base pairs and their neighbours bases.
	Functions $\gamma_2,\gamma_3,\gamma_4$ have constant ($\mathcal{O}(1)$) precision at fixed temperature~\cite{mathews1999expanded, dirks2003partition}. Typically, the function $\gamma_1$ is a logarithmic function in interior loop size $m$, using the  Jacobson-Stockmayer entropy extrapolation equation~\cite{mathews1999expanded, santa, jacobson1950intramolecular}. 	 
	But  from the Lindemann–Weierstrass theorem \cite{wiki:NT} in number theory, the logarithm to the base 2 is irrational for most natural numbers, which leads to infinite precision if we seek an approximation-free computation, which is naively impossible on finite-precision computers. 
	In this work we assume  \textbf{logarithmic precision} for interior loops, that is, $\log \text{poly}(n)$ number of bits. This assumption is  compatible 
	with software packages that implement versions of the NN model~\cite{NUPACK,viennaRNA,mfold}. 
	 
	\item \textbf{Hairpin loops}: 
    A hairpin loop is defined by a closing base pair $(i,j)$, has length  $m=j-i-1$. For $m \geq 5$, the free energy contribution is approximated using the following formula~\cite{santa}:
    \begin{equation*}\label{eq:hairpin}
    \Delta G^\text{hairpin}_{i,j} = \eta_1(m) + \eta_2(i,j,i+1,j-1)
    \end{equation*}

    Where $\eta_2$ is a function of paired bases and their neighbours and has  constant precision, and  $\eta_1$ has constant precision for $5 \leq m \leq 9$, and logarithmic in $m$ for $m > 9$~\cite{mathews1999expanded}. 

     For the special cases of $m \leq 4$:  
     Hairpin loops of length $3$ or $4$ have (temperature-dependent) free energy expressions with constant precision~\cite{santa}, partially 
     because at those lengths 
     certain sequences are particularly stable. Hairpin loops with $m \leq 2 $ are geometrically prohibited~\cite{santa}.  As with interior loops, we assume  \textbf{logarithmic precision}, $\log \text{poly}(n)$ number of bits for hairpin loop free energy.

\item \textbf{Multiloops}:
A multiloop is defined by three or more base closing pairs~\cite{dirks2007thermodynamic,fornace2020unified}, also called multibranched loops~\cite{mathews1999expanded,santa}.
Unlike other loops with single-stranded regions, multiloop free energy is defined (approximated) by a linear function~\cite{dirks2007thermodynamic,wuchty1999complete,mathews1999expanded}:\footnote{This linear form is more amenable to dynamic programming for the PF and MFE problems than the logarithmic formulae used for other loops (indeed it is an open problem to have fast algorithms for log-cost multiloops for the NN model in its full generality including rotational symmetry)~\cite{dirks2007thermodynamic,ShalabyWoods}.}
\begin{equation*} \label{eq:multi}
		\Delta G^ \textrm{multi}  = \Delta G_\textrm{init}^\textrm{multi} + b \Delta G_\textrm{bp}^\textrm{multi} + n\Delta G_\textrm{nt}^\textrm{multi}
\end{equation*}
where $\Delta G_\textrm{init}^\textrm{multi}$ is the penalty for the multiloop formation, $\Delta G_\textrm{bp}^\textrm{multi}$ is the penalty for each of its $b$ base pairs that border the interior of the multiloop, and $\Delta G_\textrm{nt}^\textrm{multi}$ is the penalty for each of the $n$ free bases inside the multiloop.
We note that $\Delta G^\textrm{multi}_\textrm{init}$, $\Delta G_\textrm{bp}^\textrm{multi}$, and $\Delta G_\textrm{nt}^\textrm{multi}$ are constants with constant precision, hence we also assume \textbf{logarithmic precision}, $\log \text{poly}(n)$ number of bits, to represent the energy contribution of multiloops. 
		
	\item \textbf{Stacks}: A stack consists fo two consecutive base pairs with no unpaired bases between them. Since there are only 4 DNA/RNA bases, there are only a constant number of 16 stacks. Hence, special energy expressions with {\bf constant precision}, and values parameterised from the literature~\cite{santa,mathews1999expanded}  are used for stacks.
	
    \item  \textbf{Bulges}: A bulge consists of two base pairs separated by unpaired bases on only one side of the loop. Special energy expressions with constant precision are used for bulges of size $m \leq 6 $. For bulges of size $m > 6$ nucleotides, the logarithmic approximation using Jacobson-Stockmayer entropy extrapolation equation~\cite{mathews1999expanded, santa, jacobson1950intramolecular} is used. Hence, we assume \textbf{logarithmic precision}, $\log \text{poly}(n)$ number of bits to represent the energy contribution of bulges.
    
    	\item \textbf{Exterior loops}: An exterior loop contains a nick (see \cref{fig:polymer}) between strands and any
    	number of closing base pairs. Exterior loops do not contribute to the free energy
    	because by definition~\cite{dirks2007thermodynamic}, $\Delta G^\text{exterior} = 0$.
\end{itemize}

\subsubsection{Simple algorithm to compute a set of candidate energy levels for an instance of NN}\label{sec:brute force EL}

Here, we show there is a simple polynomial time algorithm, in the number of bases $n$, to compute a polynomial-sized superset of the set of occupied energy levels. 
Later, in \cref{sec:more precise EL algo} we give a more efficient, but more complicated dynamic programming algorithm that gives a tighter superset of energy levels.

\begin{lemma}\label{lem:CandidateEnergyLevels}
Given a multiset of strands $s$ and temperature $T$, 
for each of the BPS, \UBP and NN models 
there is a set of candidate energy levels $\mathcal{G}^\mathcal{M}_s$ of polynomial size  and computable in polynomial time in the number of bases $n$.
More precisely, the size $|\mathcal{G}^\mathcal{M}_s|$ is linear for BPS and \UBP, and $O(n^{O(1)})$ for NN. 
\end{lemma}
\begin{proof}
As noted above, $\mathcal{G}_s^\text{\UBP} = \mathcal{G}_s^\text{BPS} = \{i \in \mathbb{N}: i \leq \lfloor n/2 \rfloor\}$ proving the statement for \UBP and BPS. 

For the NN model, the secondary structure free energy is a sum of $O(n)$ loop free energies. 
A loop involving $k$ bases has free energy $O(k)$ since loops types have constant, logarithmic or, for multiloops, linear free energy.  
Hence the per-base contribution of a loop is merely $O(1)$, implying the MFE of the $n$-base system is $O(n)$. 

The per-loop free energy has at most logarithmic precision (can be written in a logarithmic number of bits as discussed for each loop type earlier in \cref{App:candlev}). 
The free energy of a secondary structure is a sum of $O(n)$ loop free energies, 
each of at most logarithmic precision hence the precision $p$ after the decimal place of the secondary structure free energy is logarithmic $p = O(\log n)$.

The smallest value representable with this precision is 
$2^{-p} = 2^{- O(\log n)} = 2^{ - O(1) \log_2 n} = n^{- O(1)}$. 
Since the MFE is $O(n)$ the number of levels $|\mathcal{G}^{\mathrm{NN}}_s|$  satisfies  $n^{- O(1)} \cdot |\mathcal{G}^{\mathrm{NN}}_s| = O(n)$, hence $|\mathcal{G}^{\mathrm{NN}}_s|=n^{O(1)}$. 
Furthermore, $\mathcal{G}^{\mathrm{NN}}_s$ is easily  computable in time  $n^{O(1)}$ by brute force. 
\end{proof}

\subsubsection{Algorithm to compute a somewhat precise set of candidate energy levels for an instance of NN}\label{sec:more precise EL algo}

Although the algorithm from \cref{lem:CandidateEnergyLevels} is simple to understand, it may be somewhat unsatisfactory as many of the  energy levels it enumerates are unlikely to have secondary structures.
In this section, we give another algorithm for computing a set of candidate energy levels for any multi-stranded system in the NN model.  
We give two variants of the algorithm. 

The first is precise in two ways, namely (a) it is sequence-dependent and (b) computes only {\em occupied} energy levels,
and imprecise in the sense that it ignores the rotational symmetry penalty term $+k_\text{B} T \log R$ (see \cref{def:NN}).
 More precisely, for any strand multiset $s$ with $n$ bases in total, ignoring the rotational symmetry penalty, our algorithm only computes the set of {\em  occupied} energy levels---i.e. levels that have at least one secondary stricture $S$ with that free energy. 
 More formally, our algorithm computes the following set: $\{ \Delta G^\text{NN}(S) - k_\mathrm{B} T \log R \mid S \in \Omega_s  \}$.
Our algorithm is a direct translation of the multi-stranded partition function algorithm from Dirks et al.~\cite{dirks2007thermodynamic,fornace2020unified}, which we rely on for correctness, by modifying the underlying algebra from $(+,\cdot)$ to $(\cup,+)$\footnote{Note that $+$ here denotes the sumset operation~\cite{wiki:setsum} defined as follows: for any two sets $A$ and $B$, $A+B = \{a+b: a \in A, b\in B\}$.} while keeping the recursive structure as it is.  
The techniques for handling rotational symmetry for PF~\cite{dirks2007thermodynamic} and MFE~\cite{ShalabyWoods} seem to not immediately carry over, hence we leave handling of rotational symmetry as an open question.

The second variant of the algorithm allows for handing a superset of occupied energy levels, including the rotational symmetry penalty. 
The trade-off with the first variant is that the second algorithm typically returns some extra, unoccupied, energy levels.

\begin{figure}[htbp]
	\centering
	\begin{minipage}[b]{1\textwidth}
		\centering
		\includegraphics[width=1\linewidth]{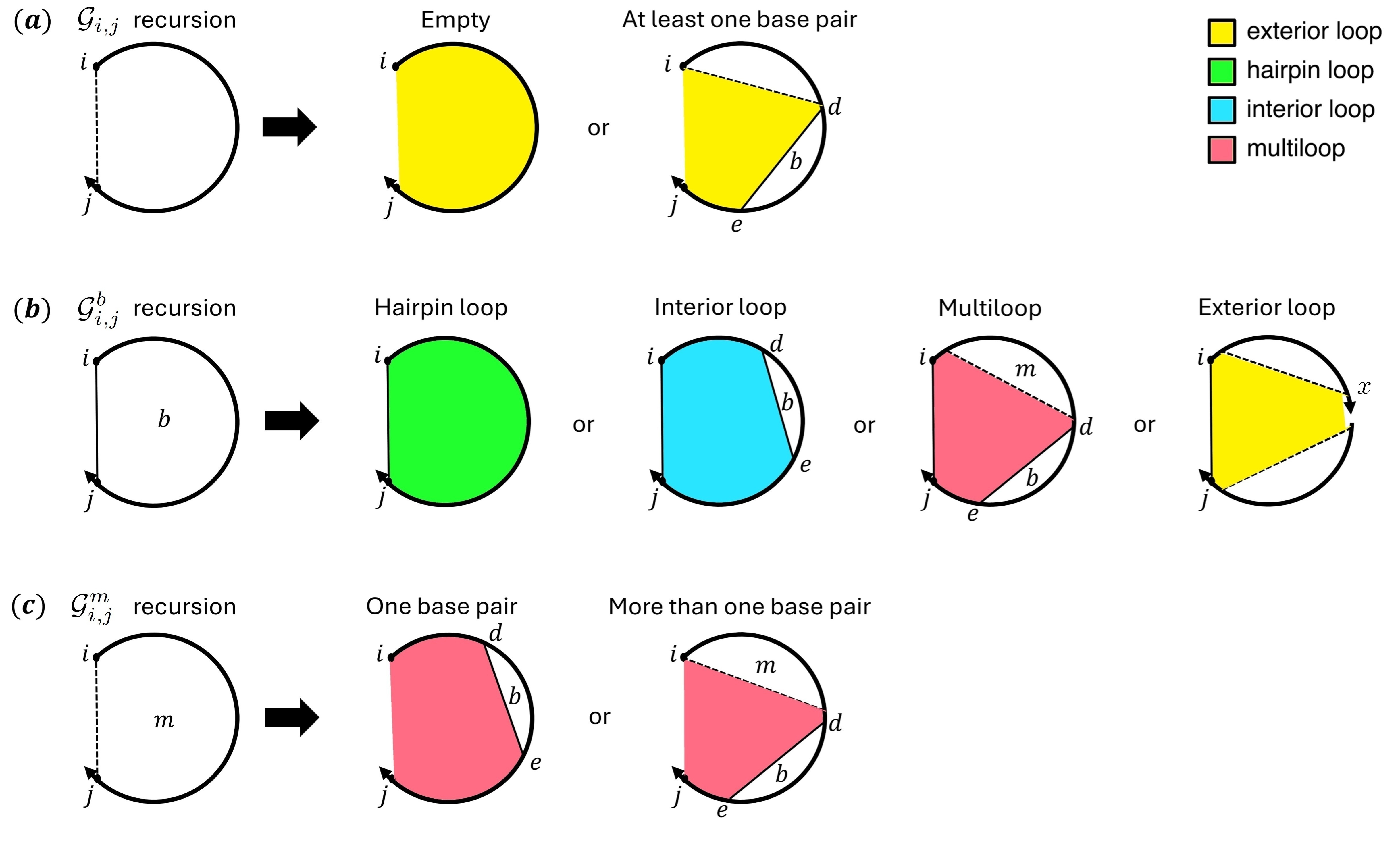} 
	\end{minipage}

	\begin{minipage}[b]{1\textwidth}
		\centering
		
		\begin{align}
			\mathcal{G}_{i,j} &= \{0\} \;\cup\; \left [ \bigcup_{i \le d < e \le j}
			\left( \mathcal{G}_{i,d-1} + \mathcal{G}^b_{d,e} \right) \right ] \label{eq:1}
			\\[6pt] 
			\mathcal{G}^b_{i,j} &= \{\Delta G^{\mathrm{hairpin}}_{i,j}\}
			\;\cup\; \left [ \bigcup_{i < d < e < j}
			\left( \mathcal{G}^b_{d,e} + \Delta G^{\mathrm{interior}}_{i,d,e,j} \right) \right ] \notag \\
			&\quad \cup\; \left [ \bigcup_{i < d < e < j}
			\Big[
			\left( \mathcal{G}^m_{i+1,d-1} + \mathcal{G}^b_{d,e} \right)
			+ \left( \Delta G^{\mathrm{multi}}_{\mathrm{init}} 
			+ 2\Delta G^{\mathrm{multi}}_{\mathrm{bp}} 
			+ \Delta G^{\mathrm{multi}}_{e+1,j-1} \right)
			\Big] \right ] \notag \\
			&\quad \cup\; \left [ \bigcup_{\substack{i \le x < j \\ x + \tfrac12 \text{ is a nick}}}
			\left( \mathcal{G}_{i+1,x} + \mathcal{G}_{x+1,j-1} \right) \right ] \label{eq:2}
			\\[6pt]
			\mathcal{G}^m_{i,j} &=
			\bigcup_{i \le d < e \le j}
			\Big[ \mathcal{G}^b_{d,e} 
			+ \left( \Delta G^{\mathrm{multi}}_{\mathrm{bp}} 
			+ \Delta G^{\mathrm{multi}}_{i,d-1} 
			+ \Delta G^{\mathrm{multi}}_{e+1,j} \right) \Big] \notag \\
			&\quad \cup\;
			\left [ \bigcup_{i \le d < e \le j}
			\Big[ 
			\left( \mathcal{G}^m_{i,d-1} + \mathcal{G}^b_{d,e} \right)
			+ \left( \Delta G^{\mathrm{multi}}_{\mathrm{bp}} 
			+ \Delta G^{\mathrm{multi}}_{e+1,j} \right)
			\Big] \right ] \label{eq:3}
		\end{align}

	\end{minipage}

\caption{Multi-stranded algorithm for computing the set of occupied energy levels, ignoring rotational symmetry, described by recursion diagrams (top) and recursion equations (bottom). Pseudo-code implementation details are provided in \cref{algo:3}.}
\label{fig:occ}
\end{figure}

\begin{lemma}\label{lem:CandidateEnergyLevelsNN}
	Given a multiset of strands $s$ and temperature $T$ for the NN model, 
	there is an $\mathcal{O}(n^4 |\mathcal{G}'|^2)$ algorithm to compute a set of candidate energy levels $\mathcal{G}^\text{NN}_s$, where $|\mathcal{G}'| = n^{O(1)}$ is defined in the proof.
	The computed set is precisely the  occupied energy levels (that have at least one structure) if we ignore rotational symmetry. 
	A second algorithm allows for rotational symmetry, but gives a (typically strict) superset of the occupied energy levels. 
\end{lemma}

\begin{proof}[Proof sketch]
	\cref{fig:occ} shows the recursion diagrams of our proposed \cref{algo:3}, and the recursion
	equations. 
	By convention~\cite{dirks2007thermodynamic,fornace2020unified, ShalabyWoods, mccaskill1990equilibrium}, a solid straight line indicates a base pair, and a dashed line indicates a region without implying that the connected bases are paired or not. Colored regions correspond to loop free energies that are explicitly incorporated at the current
	level of recursion. We compute three different types of sets of occupied energy levels, ignoring rotational symmetry, namely $\mathcal{G}_{i,j}$, $\mathcal{G}^b_{i,j}$, and $\mathcal{G}^m_{i,j}$ as follows:

	\begin{enumerate}[label=(\alph*)]
		\item $\mathcal{G}_{i,j}$ represents the set of occupied energy levels for subsequence $[i, j]$. 
		There are two cases: 
		\begin{enumerate}[label=(\arabic*)]
			\item   Either there are no base pairs corresponding to an exterior loop with free energy $\Delta G^\text{exterior} = 0$, and hence a contribution of $\{0\}$ is unioned into its corresponding set of occupied energy levels $\mathcal{G}_{i,j}$, see \cref{eq:1}.
			
			\item Else there is a $3'$-most base pair $(d,e)$. In this case, computing the set of candidate energy levels makes use of previously computed subsequence set of occupied energy levels $\mathcal{G}_{i,d-1}^b$ and $\mathcal{G}_{d,e}^b$. Using the {\em sumset}~\cite{wiki:setsum},  $\mathcal{G}_{i,d-1} + \mathcal{G}_{d,e}^b = \{x+y: x \in \mathcal{G}_{i,d-1}, y \in \mathcal{G}_{d,e}^b \}$ to compute the contribution of that specific $3'$-most base pair $(d,e)$ to $\mathcal{G}_{i,j}$. The union over all possible choices for $3'$-most base pair $(d,e)$ yields $\mathcal{G}_{i,j}$ at the end, see \cref{eq:1}. Note that we make use of the independence of loop contributions in the NN energy model which in turn relies on the assumption of disallowing pseudoknots.
		\end{enumerate}
		
		\item  $\mathcal{G}_{i,j}^b$ represents the set of occupied energy levels for subsequence $[i, j]$ with the restriction that bases $i$ and $j$ are paired. 
		There are four cases: 
		
		\begin{enumerate}[label=(\arabic*)]
			\item There are no additional base pairs, which corresponds to a hairpin loop, and hence a contribution of $\{\Delta G^\text{hairpin}_{i,j}\}$ is unioned into its corresponding set of occupied energy levels $\mathcal{G}^b_{i,j}$, see \cref{eq:2}. 
			
			\item There is exactly
			one additional base pair $(d,e)$ that corresponds to an interior loop, and hence a contribution of $ \mathcal{G}_{d,e}^b + \Delta G^\text{interior}_{i,d,e,j} = \{x+\Delta G^\text{interior}_{i,d,e,j}: x \in \mathcal{G}_{d,e}^b\}$ to $\mathcal{G}^b_{i,j}$, see \cref{eq:2}.  
			
			\item There is more than one
			additional base pair, corresponding to a multiloop with $3'$-most base pair $(d,e)$, and at least one additional base pair specified in a previously computed subsequence set of occupied energy levels $\mathcal{G}^m_{i+1,d-1}$. Hence a contribution of $\Big[\big(\mathcal{G}^m_{i+1,d-1} +  \mathcal{G}^b_{d,e}\big) + \big(\Delta G^\mathrm{multi}_\mathrm{init} + 2 \Delta G^\mathrm{multi}_\mathrm{bp} + \Delta G^\mathrm{multi}_{e+1,j-1}\big) \Big]$ of that specific $3'$-most base pair $(d,e)$ to $\mathcal{G}^b_{i,j}$, which includes the multiloop contribution, see \cref{eq:2}. Note that $\Delta G^\mathrm{multi}_{i,j} = (j-1+1) \Delta G^\mathrm{multi}_\text{nt}$. 
			
			\item There is an exterior loop
			containing a nick after nucleotide $x$. Hence a contribution of $\mathcal{G}_{i+1,x} +  \mathcal{G}_{x+1,j-1}$ to $\mathcal{G}^b_{i,j}$ of that specific nick after $x$, see \cref{eq:2}.
		\end{enumerate}
		
		\item  $\mathcal{G}_{i,j}^m$ represents the set of occupied energy levels for subsequence $[i, j]$ with the restriction that  $[i, j]$  is inside a multiloop and contains at least one base pair. There are two cases: 
		
		\begin{enumerate}[label=(\arabic*)]
			
			\item There is exactly
			one additional base pair $(d,e)$ defining the multiloop, and hence a contribution of $\Big[ \mathcal{G}^b_{d,e} + \big(\Delta G^\mathrm{multi}_\mathrm{bp} + \Delta G^\mathrm{multi}_{i,d-1} + \Delta G^\mathrm{multi}_{e+1,j} \big) \Big]$ to $\mathcal{G}^m_{i,j}$, see \cref{eq:3}. 
			
			\item There is more than one
			additional base pair defining the multiloop with $3'$-most base pair $(d,e)$, hence a contribution of $\Big[	(\mathcal{G}^m_{i,d-1} + \mathcal{G}^b_{d,e}) + \big(\Delta G^\mathrm{multi}_\text{bp} + \Delta G^\mathrm{multi}_{e+1,j}\big) \Big]$ to $\mathcal{G}^m_{i,j}$, see \cref{eq:3}. 
			
		\end{enumerate}
	\end{enumerate}

\paragraph*{The second variant of our algorithm} \cref{algo:3} returns the set of occupied energy levels $ \mathcal{G}_{1,n} + (c-1) \Delta G^\text{assoc}$, which ignores rotational symmetry. So the last step is to include extra $\tau (v(\pi)) - 1$ levels\footnote{$\tau (v(\pi))$ is the number of divisors of the the highest degree of symmetry $v(\pi)$ for strand ordering $\pi$, see~\cite{ShalabyWoods} for more details. Note that, $\tau (v(\pi))$ is a constant since $c = \mathcal{O}(1)$ strands.} for each computed energy level to cover all rotational symmetry possibilities\footnote{Shalaby and Woods~\cite{ShalabyWoods} showed that, if $S$ is $R$-fold rotationally symmetric secondary structure, with a specific strand ordering $\pi$, then $R$ must be a divisor of $v(\pi)$.}. In other words for each energy level $l$ computed by \cref{algo:3}, we add a new energy level $l' = l + k_\text{B} T \log R$, for all $R>1$ and $R$ divides the highest degree of symmetry $v(\pi)$. This step of adding extra energy levels gives at the end  a superset of occupied energy levels $\mathcal{G}^{\mathrm{NN}}_s$, including the rotational symmetry penalty. Formally,  $\{ \Delta G^\text{NN}(S) \mid S \in \Omega_s  \} \subseteq \mathcal{G}_s^\text{NN}$. 
\newline

\paragraph*{Time complexity}   
The multi-stranded partition function algorithm from Dirks et al.~\cite{dirks2007thermodynamic,fornace2020unified} has $\mathcal{O}(n^4)$ time complexity. But modifying its underlying algebra from $(+,\cdot)$ to $(\cup,+)$ to get \cref{algo:3} changes the time complexity because at each recursive call, the sumset operation needs an extra $|\mathcal{G}'|^2$ steps in the worst case, where   $\mathcal{G}' = \max \{\mathcal{G}_{i,j} \text{ for any subsequence } [i,j]\}$. Hence, the time complexity of \cref{algo:3}  is  $\mathcal{O}(n^4 |\mathcal{G}'|^2)$. 

\end{proof}

\begin{algorithm}
	\caption{\small Pseudocode  that takes as input: $c=\mathcal{O}(1)$ strands with total number of bases (length) $n$ and strand ordering $\pi$. 
		Runs in  $\mathcal{O}(n^4 |\mathcal{G}'|^2)$ time with
		recursive calls illustrated in \cref{fig:occ}. Here $\mathcal{G}' = \max \{\mathcal{G}_{i,j} \text{ for any subsequence } [i,j]\}$. Nicks between strands are denoted by half indices (e.g.~$x+ \frac{1}{2}$). 
		The function $\eta[i+ \frac{1}{2}, j+\frac{1}{2}]$ returns the number of nicks in the interval $[i+ \frac{1}{2}, j+\frac{1}{2}]$. 
		The shorthand $\eta[i+ \frac{1}{2}]$ is equivalent to $\eta[i+ \frac{1}{2}, i+\frac{1}{2}]$ and by convention, $\eta[i+ \frac{1}{2}, i-\frac{1}{2}] =0$. For any two sets $X$ and $Y$, we define $X ~ \cup \! = Y $ to be equivalent to $X = X \cup Y$.}\label{algo:3}
	\begin{algorithmic}[1]
		\footnotesize
		\State Initialize $\mathcal{G}, \mathcal{G}^b, \mathcal{G}^m$ by setting all values to $\Phi$, except $\mathcal{G}_{i,i-1} = \{0\}$ for all $i=1,\ldots,n$. For any set $X$, we define $X \cup \Phi = \Phi$.
		
		\For{$l \gets 1 \ldots n$}
		\For{$i \gets 1 \ldots n-l+1$}
		\State $j = i+l-1$
		
		\Comment{$\mathcal{G}^b$ recursion equations} 
		\If{$\eta[i+\frac{1}{2}, j-\frac{1}{2}] ==0$}
	\State 	$\mathcal{G}_{i,j}^b =\{\Delta G_{i,j}^\text{hairpin}\}$
		\Comment{hairpin loop requires no nicks}

		\EndIf 
		
		\For{$d \gets i+1 \ldots j-2$}
		\Comment{loop over all possible $3'$-most pairs $(d,e)$} 
		
		\For{$e \gets d+1 \ldots j-1$} 
		
		\If{$\eta[i+\frac{1}{2}, d-\frac{1}{2}] ==0$ and $\eta[e+\frac{1}{2}, j-\frac{1}{2}] ==0$}
		
		\State 	$\mathcal{G}_{i,j}^b ~ \cup \! = (\mathcal{G}^b_{d,e} + \Delta G^\mathrm{interior}_{i,d,e,j} )$

		\EndIf 
		
		\If{$\eta[e+\frac{1}{2}, j-\frac{1}{2}] ==0$ and $\eta[i+\frac{1}{2}] ==0$ and $\eta[d-\frac{1}{2}] ==0$}
		\Comment{multiloop: no nicks} 
		
		\State  $\mathcal{G}_{i,j}^b ~ \cup\!= \Big[(\mathcal{G}^m_{i+1,d-1} +  \mathcal{G}^b_{d,e}) + (\Delta G^\mathrm{multi}_\mathrm{init} + 2 \Delta G^\mathrm{multi}_\mathrm{bp} + (j-e-1) \Delta G^\mathrm{multi}_\text{nt}) \Big]$

		\EndIf
		\EndFor		
		\EndFor
		
		\For{$x \in \{i, \ldots,j-1\}$ s.t. $\eta[x+\frac{1}{2}] = 1$} 
		\Comment{loop over all nicks $\in [i+\frac{1}{2}, j-\frac{1}{2}]$} 
		
		\If{($\eta[i+\frac{1}{2}] == 0$ and $\eta[j-\frac{1}{2}] == 0$) or ($i==j-1$) or  ($x==i$ and $\eta[j-\frac{1}{2}] == 0$) or
			\par \hskip\algorithmicindent  
			\hspace{6 mm} ($x==j-1$ and $\eta[i+\frac{1}{2}] == 0$)}
		
		\State  $\mathcal{G}_{i,j}^b ~\cup \!= \big[\mathcal{G}_{i+1,x} +  \mathcal{G}_{x+1,j-1} \big]$\Comment{exterior loops} 
		
		\EndIf
		\EndFor
		
		\Comment{$\mathcal{G}, \mathcal{G}^m$ recursion equations} 
		
		\If{$\eta[i+\frac{1}{2}, j-\frac{1}{2}] == 0$} 
		$\mathcal{G}_{i,j} = \{0\}$\Comment{empty substructure} 
		
		\EndIf
		
		\For{$d \gets i \ldots j-1$} \Comment{loop over all possible $3'$-most pairs $(d,e)$}
		
		\For{$e \gets d+1 \ldots j$}
		
		\If{$\eta[e+\frac{1}{2}, j-\frac{1}{2}] == 0$} 
		
		\If{$\eta[d-\frac{1}{2}] == 0$ or $d==i$}
		
		\State $\mathcal{G}_{i,j} ~ \cup \! =  \big[\mathcal{G}_{i,d-1} + \mathcal{G}^b_{d,e}\big]$
		\EndIf
		
		\If{$\eta[i+\frac{1}{2}, d-\frac{1}{2}] == 0$}
		
		\State $\mathcal{G}_{i,j}^m   ~ \cup \!    = 
		\Big[ \mathcal{G}^b_{d,e} + \big(\Delta G^\mathrm{multi}_\mathrm{bp} + (d-i + j-e) \Delta G^\mathrm{multi}_\text{nt} \big) \Big]$
		\Comment{single base pair (bp)}
		
		\EndIf
		\If{$\eta[d-\frac{1}{2}] == 0$}
		\State $\mathcal{G}_{i,j}^m ~ \cup \! = \Big[	(\mathcal{G}^m_{i,d-1} + \mathcal{G}^b_{d,e}) + \big(\Delta G^\mathrm{multi}_\text{bp} + (j-e) \Delta G^\mathrm{multi}_\text{nt}\big) \Big]$
		\Comment{more than 1 bp}

		\EndIf
		
		\EndIf
		
		\EndFor
		
		\EndFor
		
		\EndFor
		\EndFor 
		\State \Return $\mathcal{G}_{1,n} + (c-1) \Delta G^\text{assoc}$		
		
	\end{algorithmic}
\end{algorithm}